\documentclass[10pt,a4paper]{article}
\usepackage[utf8]{inputenc}
        \usepackage{geometry}
\usepackage{authblk}
\usepackage{cite}
\usepackage{amsmath,amssymb,amsfonts,amsthm}
\usepackage{algorithmic}
\usepackage{graphicx}
\usepackage{textcomp}
\usepackage{xcolor}
\usepackage{mdframed}
\usepackage{boxedminipage}
\usepackage{caption}
\usepackage{subcaption}
\usepackage{float}
\usepackage{enumerate}
\usepackage{hyperref}
\def\BibTeX{{\rm B\kern-.05em{\sc i\kern-.025em b}\kern-.08em
    T\kern-.1667em\lower.7ex\hbox{E}\kern-.125emX}}
\usepackage{enumitem}
\setlist{topsep=2pt, itemsep=1pt}

\newcommand{\mc}{\mathcal}
\newcommand{\mbf}{\mathbf}
\newcommand{\msf}{\mathsf}

\newcommand{\ignore}[1]{}

\newcommand{\protocol}{\Pi_\mathsf{fruit}^{\mathtt{H},\mathsf{d}}(p_b,p_f,r)}
\newcommand{\block}{\mathbf{B}}
\newcommand{\fruit}{\mathbf{f}}
\newcommand{\pb}{p_b}
\newcommand{\pf}{p_f}

\newcommand{\rf}{R_f}
\newcommand{\recency}{r}
\newcommand{\nonce}{\eta}
\newcommand{\digest}{\mathsf{d}}
\newcommand{\record}{\mathsf{m}}
\newcommand{\tx}{\mathsf{tx}}
\newcommand{\previous}{h_{-1}}
\newcommand{\chain}{\mathit{chain}}

\newcommand{\oracle}{\mathtt{H}}

\newcommand{\exec}[1]{\mathcal{E}_{\mathcal{Z},\mathcal{A}_{#1}}}
\newcommand{\honexec}{\mathcal{E}_{\mathcal{Z},\mathcal{H}_\mathbf{C}}}
\newcommand{\umax}{U^\mathsf{max}_\mathbf{C}}
\newcommand{\umin}{U^\mathsf{min}_\mathbf{C}}
\newcommand{\rounds}{N}
\newcommand{\upper}{t}

\newcommand{\secpar}{\kappa}

\newcommand{\party}[1]{P_#1}
\newcommand{\honest}{\mathbf{H}}
\newcommand{\corrupt}{\mathbf{C}}
\newcommand{\innocent}{\mathcal{H}_{\corrupt}}
\newcommand{\cost}{C}
\newcommand{\reward}{R_f}

\newcommand{\protocolevp}{\Pi_\mathsf{fruit}^{\mathcal{O}_{\msf{lc}/\msf{fs}/\msf{tx}/\msf{ro}}}(p_b,p_f)}
\newcommand{\negl}{\mathsf{negl}}

\newcommand{\single}{\Pi_\mathsf{single}^{\mathcal{O}_{\msf{lc}/\msf{fs}/\msf{tx}/\msf{ro}/\msf{ltx}}}(p_b,p_f)}

\newenvironment{thickframe}{\begin{center}\begin{boxedminipage}{0.98\textwidth}}{\end{boxedminipage}\end{center}}

\newtheorem{theorem}{Theorem}

\newtheorem{definition}{Definition}
\theoremstyle{definition}
\newtheorem{remark}{Remark}
\theoremstyle{theorem}

\author[1]{Aikaterini-Panagiota Stouka}
\author[2]{Thomas Zacharias}

\affil[1]{Nethermind, UK}
\affil[2]{The University of Edinburgh, UK}
 \date{}
\begin{document}
\newtheorem{claim}{Claim}[theorem]

\title{On the (De)centralization of FruitChains}

\maketitle

\begin{abstract}
One of the most important features of blockchain protocols is
decentralization, as their main contribution is that they formulate a distributed ledger that will be maintained and extended without the need of a trusted party. Bitcoin has been criticized for its tendency to centralization, as very few pools control the majority of the hashing power. Pass et al. proposed FruitChain [PODC 17] and claimed that this blockchain protocol mitigates the formation of pools by reducing the variance of the rewards in the same way as mining pools, but in a fully decentralized fashion. Many follow up papers consider that the problem of centralization in Proof-of-Work (PoW) blockchain systems can be solved via lower rewards' variance, and that in FruitChain the formation of pools is unnecessary. 

Contrary to the common perception, in this work, we prove that lower variance of the rewards does not eliminate the tendency of the PoW blockchain protocols to centralization; miners have also other incentives to create large pools, and specifically to share the cost of creating the instance they need to solve the PoW puzzle.

We abstract the procedures of FruitChain as oracles and assign to each of them a cost. Then, we provide a formal definition of a pool in a blockchain
system, and by utilizing the notion of equilibrium with virtual payoffs (EVP) [AFT 21], we prove that there is a completely centralized EVP, where all the parties form a single pool controlled by one party called the pool leader. The pool leader is responsible for creating the instance used for the PoW procedure. To the best of our knowledge, this
is the first work that examines the construction of mining
pools in the FruitChain system. \smallskip

\end{abstract}

    \section{Introduction}\label{sec:introduction}
Bitcoin introduced by Nakamoto \cite{nakamoto} is the first established decentralized cryptocurrency. The Bitcoin blockchain protocol formulates a ledger that consists of a chain of blocks that include transactions. This ledger is maintained without the need of a trusted party and is extended by peer to peer nodes called \emph{miners}. Each miner extends the chain when it manages to solve a Proof-of-Work (PoW) puzzle \cite{puzzle,puzzle2,puzzle3,puzzle4} using computational/hashing power. If more than one chains have been formed, the longest of them constitutes the ledger. The economic incentives of the miners to participate in the mining process are (i) the newly-minted rewards they earn when they produce a block that extends the ledger, and (ii) the transaction fees  which constitute the ``tip'' for the miner who includes a transaction in its block.

Although Bitcoin has been characterized as the biggest financial innovation of the fourth industrial revolution \cite{LI2021120383}, it has been criticized for various reasons, such as its vulnerability to \emph{selfish mining} attacks \cite{selfish} and its tendency to centralization \cite{DBLP:journals/corr/abs-1811-08572, gap, DBLP:journals/corr/abs-1904-02368,6824541}. In more detail, when a malicious attacker performs a selfish mining attack, it reduces the fraction of the blocks in the ledger that belong to the honest miners (the miners that follow the Bitcoin protocol). As far as  Bitcoin centralization is concerned, miners are organized into mining pools. Currently, only four pools constitute the majority of the computational power \footnote{\url{https://btc.com/stats/pool.}}. The miners who join pools solve computational puzzles of lower difficulty (partial PoW), they get paid regularly according to the pool rules, and their rewards have lower variance compared to solo mining~\cite{Rosenfeld2011AnalysisOB, 7163020, https://doi.org/10.48550/arxiv.1905.05999, 10.1007/978-3-662-54970-4_28}. 
\par In order to prevent selfish mining attacks and achieve \emph{fairness}, Pass et al.~\cite{fruitchain} propose a blockchain protocol called \emph{FruitChain} that uses the 2-for-1 PoW technique (introduced in \cite{backbone}). According to \cite{fruitchain}, a protocol satisfies fairness  when
with overwhelming probability, in every long enough segment of the ledger, any honest set of parties is guaranteed to hold a fraction of blocks that is very close to its relative computational power. In addition, \cite{fruitchain} states that the FruitChain protocol could reduce the variance of the rewards similarly to mining pools, but in a ``fully decentralized way''. Namely, in FruitChain, the parties can produce via mining either blocks or \emph{fruits}, where fruits have much lower difficulty than blocks and can play the role of partial PoW in mining pools. To this direction, many follow up papers correlate (FruitChain protocol's)  decentralization with reducing the variance of the rewards. Some notable examples of such works are the following: (i)  \cite{236314} recognizes high variance as the main motivation for joining a pool; (ii) \cite{9496180} states that mining pools are unnecessary, because miners can produce fruits in short time; (iii) according to \cite{Huang_2021}, when the parties' rewards are concentrated with high probability to their initial resources, the parties lose their motivation to form mining pools; (iv) \cite{Sarkar2020ANB} states that a distribution of block rewards that is equitable makes the formation of mining pools redundant; (v) \cite{kellerconsensus} states that there is no mean of pool existence in Fruitchain; (vi) according to \cite{10.1145/3449301.3449335}, as the fruit can be mined in a very short period, pools are not necessary. Other works  that correlate the incentives of the parties to form/join a pool with the variance of their rewards are  \cite{Tanniru2021,article, journals/iacr/FitziGKR18, 10.1145/3318041.3355458,9724503,kellerconsensus,https://doi.org/10.48550/arxiv.2104.01918}.\\[2pt]
\indent\emph{Our results.}
In this paper, we revisit the decentralization of the FruitChain protocol \cite{fruitchain} and argue that, contrary to the common perception, \emph{lower variance of the rewards does not eliminate the tendency of PoW blockchain protocols to centralization}. In particular, we focus on another motivation of the miners to form pools, which is to share the cost of creating the \emph{instance} (i.e., the block header that the miners iteratively hash applying different nonce in each iteration) they need to solve the PoW puzzle. By utilizing a notion of equilibrium called \emph{equilibrium with virtual payoffs} (EVP) presented in \cite{EVP}, we prove that in FruitChain, there is a completely centralized EVP, where all the parties form a single pool controlled by a single party (pool leader) responsible for creating the instance used for the PoW task. In more detail, 
\begin{enumerate}[itemsep=2pt]
\item We abstract the procedures of the FruitChain protocol as oracles and we assign a cost to each procedure. 

\item We provide a formal definition of a pool in a blockchain system. We treat the pool as a description of a subset of parties, along with their communication setting and their execution guidelines. Although our definition is generic, we focus on PoW pools where the collective rewards of the pool are shared among its members in ``off-chain'' manner, i.e., not enforced by the underlying blockchain. 

\item We introduce a set of rules of a completely centralized pool that includes all the miners in the FruitChain system. In this single pool, only the pool leader decides which chain constitutes the ledger and the contents of the blocks that will extend this chain, but all the members (including itself) share the cost of this procedure. To the best of our knowledge, this is the first work that examines the construction of mining pools in the FruitChain system.

\item We prove that in the FruitChain system, joining this completely centralized pool and following its rules is an EVP. Note that we are able to prove this statement because the rules of this pool disincentivize the members of the pool from (i) attacking each other and (ii) making the betrayal of the pool more profitable than sticking to the pool and sharing the costs.

\item We discuss possible directions to mitigate the tendency to centralization due to sharing verification costs.
\end{enumerate}
\indent\emph{Other related work}.
In \cite{EVP}, Kiayias et al. present the EVP notion that we use in our results. As far as decentralization is concerned, the authors prove that in Bitcoin, forming a single dictatorial pool is EVP, when the cost of processing transactions is not negligible. In our work, we show that although the FruitChain protocol offers lower variance in the miners' rewards, the centralization problem is not solved. Apart from deploying the EVP framework to study the (de)centralization of another protocol, our formal treatment extends and improves the one of~\cite{EVP} w.r.t. the following aspects:
\begin{itemize}
    \item We take into account not only the cost of processing transactions and the cost of making random oracle queries (as \cite{EVP} does), but also the cost of finding the longest chain and the cost for finding all the fruits whose digest should be included in the instance needed for mining.
    \item In \cite{EVP}, it is assumed that when a pool member ignores the instance of the pool leader and mines on a different instance, the pool leader can detect this deviant behavior and remove the member from the pool. This assumption excludes strategies where the deviant member mines under a different identity (public key) than the one the pool leader has recorded. In our case, we address this by having the pool leader dissolve the pool whenever it detects the aforementioned deviant behavior.
\end{itemize}

In \cite{9230398}, the authors prove that if a ``fair'' reward function is used in their setting, there is no equilibrium with more than one pool. Note that the setting in~\cite{9230398} (i) refers to proof of stake protocols, which means that there is no cost for mining, and (ii) is deterministic, so it cannot  capture the probabilistic nature of  the FruitChain protocol. 
Some other works that consider cost sharing as a possible reason for centralization in blockchain protocols are \cite{breidenbach2021chainlink,Natoli2019DeconstructingBA, 10.1145/3318041.3355470}. Other works related to decentralization in blockchain protocols are \cite{Gencer2018DecentralizationIB, 10.1145/3318041.3355463,azouvi2021levels,https://doi.org/10.48550/arxiv.2112.09941,9488812}.

\section{Framework}\label{sec:framework}
At this point, we will give an overview of the notion of \textit{coalition-safe equilibria with virtual payoff|} presented in \cite{EVP}. This notion generalizes the equilibrium notion presented in \cite{fruitchain}, and it is based on the execution model of \cite{backbone} and the ``real-world'' protocol execution model of   \cite{Canetti2,Canetti1,UC2,Canetti4}. It examines two executions of a blockchain protocol $\Pi$. In the first execution, all the participants follow the protocol and in the second execution, there exists a strategic coalition that deviates trying to maximize its collective utility.  

\subsection{Notation}\label{subsec:framework_notation}
We use $\secpar$ as the security parameter. We write $X\sim\mc{D}$ to denote that the random variable $X$ follows the distribution $\mc{D}$. The mean of the random variable $X$ is denoted by $E[X]$. By $\msf{Bin}(K,p)$, we denote the binomial distribution with $K$ trials and success probability $p$. We write $\negl(\secpar)$ to denote that a function is negligible in $\secpar$, i.e., asymptotically smaller than the inverse of any polynomial. By $|x|$, we denote the absolute value of $x$. We use `$||$' to denote the concatenation operation.

\subsection{Protocol Execution Model}\label{subsec:framework_execution}
The execution model for a blockchain protocol $\Pi$ comprises an environment $\mc{Z}$, an adversary $\mc{A}$ and the participating parties $\party{1}, \ldots, \party{n}$. The protocol execution is progressing in rounds. The environment reflects the external world to the protocol and  decides the number of rounds the execution will run. \emph{$\rounds$-admissible} environment will be the environment that performs the protocol $\rounds$ rounds, where $\rounds$ is a polynomial in the security parameter $\secpar$ that is higher than $\secpar$. Before the beginning of the execution the adversary chooses which parties it will control. Let $\corrupt$ be the set of parties controlled by the adversary and $\honest$ the set of the remaining parties called ``honest''. 
\par During each round, the environment gives inputs to the parties and activates them in a round-robin fashion (cf.~\cite{backbone}).   The parties that belong to $\honest$ follow the protocol and the parties that belong to $\corrupt$ follow the instructions of the adversary. Note that the adversary reflects a \emph{strategic coalition} that deviates from the protocol in a way that maximizes its collective utility (the sum of the utilities of all the parties that belong to $\honest$).
\par The communication between the parties is controlled by a functionality called \textit{Diffuse functionality} defined in \cite{backbone}. This functionality guarantees that every message sent from an honest party will be delivered to every other honest party by the end of each round. It allows though the adversary to rearrange the order of the messages during the round. This means that the adversary is allowed to deliver its messages first. The adversary that follows the protocol but rearranges the messages so that it delivers its messages first is denoted by $\innocent$. Note that the adversary may send some of its messages to a subset of the honest parties, so the honest parties at the end of a round can have different local view, thus local chain.
\par During the execution, the parties interact with a number of oracles $\mc{O}_1,\ldots,\mc{O}_{w_\Pi}$ that are protocol-specific for $\Pi$. For instance, $\mc{O}_j$ can be a random oracle, a signing oracle, a transaction validity oracle, etc. There is a limited number of queries that each party can make to each of $\mc{O}_1,\ldots,\mc{O}_{w_\Pi}$ per round that is denoted by $q_1,\ldots,q_{w_\Pi}$, respectively. In addition, we consider that each single query to $\mc{O}_1,\ldots,\mc{O}_{w_\Pi}$ has a non-zero cost denoted by $\cost_1,\ldots,\cost_{w_\Pi}$, respectively. 
\par If we fix the environment $\mc{Z}$ and the adversary $\mc{A}$, then the execution can be seen as a random variable denoted by $\exec{}$.

\subsection{Coalition-safe Equilibria with Virtual Payoff}\label{subsec:framework_equilibria}
The notion of \textit{coalition-safe equilibria with virtual payoff} examines the executions $\exec{}$ and $\honexec$ for an arbitrary environment $\mc{Z}$ and an arbitrary adversary $\mc{A}$ that corrupts a set $\corrupt$ including at most $\upper$ parties, where $\upper$ will be a parameter in this definition. In $\honexec$, the adversary follows the protocol but rearranges the messages to deliver its messages first. In $\exec{i}$ the adversary deviates arbitrarily. 
\par The notion compares the utility of the adversary in these two executions. Both executions have the same number rounds (the environment is admissible) so that  their utilities can be compared in a meaningful way. The utility of the adversary will be the sum of the utilities of the parties in $\corrupt$. Note that each honest party can have a different view on the utility of each other party and thus, the utility of the adversary. The reason is that the utility is computed based on the rewards which, in turn, are based on the parties', potentially different, local chains. The notion uses the \textit{lowest utility} of the adversary among all the honest parties' local views for the execution $\honexec$, denoted by $\umin(\honexec)$, and the \textit{highest} utility for the execution $\exec{}$, denoted by $\umax(\exec{})$.
Note that $\umin(\honexec)$ and $\umax(\exec{})$ are random variables over the coins of the adversary, the environment, the parties and the oracles.
\begin{definition}\label{def:EVP}
Let $\epsilon,\epsilon'$ be non-negative real values. A protocol $\Pi$ is \emph{$(\upper,\epsilon,\epsilon')$-equilibrium with virtual payoffs} (EVP) when for every $\rounds$-admissible environment $\mc{Z}$ and for every PPT adversary $\mc{A}$ that controls an arbitrary set $\corrupt$ of at most $\upper$ parties it holds that 
\begin{equation}\label{eq:EVP}
\umax(\exec{}) \leq \umin(\honexec) + \epsilon \cdot \mid \umin(\honexec) \mid +\epsilon'   
\end{equation}
with overwhelming (i.e., $1-\negl(\secpar)$) probability.
\end{definition}
According to Eq.~\eqref{eq:EVP}, the closer that the values of $\epsilon$ and $\epsilon'$ get to $0$, the ``tighter'' the equilibrium is. Some examples of utility functions of the adversary are: (i) absolute rewards, (ii) absolute rewards minus absolute cost (profit) and (iii) relative rewards. 
Note that if the adversary can deviate from the protocol and increase significantly its utility on the view of just one honest party with a non-negligible probability, then the protocol does not satisfy this notion.

For the rest of the paper, we refer to the framework presented in this section as the \emph{EVP framework}.

\section{The FruitChain Protocol}\label{sec:description}

In this section, we devise an adaptation of the FruitChain protocol~\cite{fruitchain} to the EVP framework of~\cite{EVP} outlined in Section~\ref{sec:framework}. In our adaptation, we take into account the cost of a random oracle query as well as the costs of deciding on a new local state, validating retrieved messages, and extracting sequences of records of transactions. Before the presentation of our adaptation, we provide an overview of the original  protocol description.

\subsection{Overview of the FruitChain Protocol}\label{subsec:overview}
In the FruitChain protocol \cite{fruitchain}, miners store transactions in \textit{fruits} instead of blocks. In order for a miner to create a fruit, it needs to perform PoW, as it does to produce blocks, yet fruit mining has lower difficulty. In more detail, the miner performs a 2-for-1 PoW procedure introduced in \cite{backbone}. In particular, the miner computes hashes of a specific input, where the prefix and the suffix of the hash determine whether a block or a fruit has been mined, respectively. Fruits are stored in blocks and they need to be \textit{recent} i.e., every fruit points to a block that is not far from the latest block of the ledger. 

At a high level, the FruitChain protocol prevents selfish mining attacks~\cite{selfish} because even if an attacker withholds a block, the fruits of this block that are still recent can be stored in a later block. The restriction of recency exists so that an attacker is not able to precompute an excessive amount of fruits and reveal it later, thus disrupting the \textit{chain quality} of the protocol\footnote{Chain quality was introduced in \cite{backbone} and is related to the fraction of the adversarial blocks in a sufficiently long segment of the ledger.}.  

\subsection{Parameterization and Basic Concepts}\label{subsec:basic}
The FruitChain PoW protocol is parameterized by:
\begin{enumerate}
    \item A \emph{random oracle} $\oracle(\cdot)$ that outputs strings of length $\geq 2\secpar$. The oracle responds to (block and fruit) mining queries.
    \item A \emph{collision resistant hash function} (CRHF) $\digest(\cdot)$, utilized to digest sets of fruits.
    \item A \emph{block mining hardness parameter} $\pb$. This is the probability that the $\oracle(\cdot)$ response leads to the successful mining of a block. 
    \item A \emph{fruit mining hardness parameter} $\pf$. This is the probability that the $\oracle(\cdot)$ response leads to the successful mining of a fruit. Probability $\pf$ is significantly greater {than $\pb$}.
    \item A \emph{recency parameter} $\recency$ that determines how far back can a fruit ``hang'', i.e., the fruit needs to point to an earlier block in the chain which is not too far from the block which records the fruit itself. 
\end{enumerate}

The structure of a valid fruit $\fruit$ is denoted by $\fruit=\langle \previous,h_f,\nonce,$ $\msf{dig},\record,h\rangle$, where 
\begin{itemize}
    \item[-] $\previous$ points to the previous block's reference.
    \item[-] $h_f$ is the \emph{pointer} of $\fruit$ to the block that $\fruit$ is hanging from.
    \item[-] $\record$ is the record to be contained in $\fruit$.
    \item[-] $\nonce$ is a random nonce denoting a solution to the computational puzzle that derives from conditions (1),(2) in Definition~\ref{def:valid_fruit}.
    \item[-] $\msf{dig}$ is the digest of some set of fruits $\mbf{F}$.
    \item[-] $h$ is the reference of $\fruit$, i.e., a hash of the previous fields.
\end{itemize}
\begin{definition}[Fruit validity]\label{def:valid_fruit}
A fruit $\fruit:=\langle \previous,h_f,\nonce,\msf{dig},\record,h\rangle$ is \emph{valid}, if the following hold:
\begin{enumerate}
    \item $\oracle(\previous||h_f||\nonce||\msf{dig}||\record)=h$.
    \item $[h]_{-\secpar}<D_{\pf}$, where $[h]_{-\secpar}$ denotes the last $\secpar$ bits of $h$, and $D_{\pf}$ is the difficulty value such that the probability that an input satisfies the relation is $\pf$.
    
\end{enumerate}
We say that a fruit set $\mbf{F}$ is \emph{valid}, if either it contains only valid fruits, or $\mbf{F}=\emptyset$.
\end{definition}

The structure of a valid block $\block$ is denoted by $\block:=\langle\langle \previous,h_f,\nonce,$ $\msf{dig},\record,h\rangle,\mbf{F}\rangle$, where 
\begin{itemize}
    \item[-] $\previous$ points to the previous block's reference.
    \item[-] $h_f$ is some fruit pointer.
    \item[-] $\record$ is the record to be contained in some fruit.
    \item[-] $\nonce$ is a random nonce denoting a solution to the computational puzzle that derives from condition (3),(4) in Definition~\ref{def:valid_block}.
    \item[-] $\msf{dig}$ is the digest to the set of fruits $\mbf{F}$.
    \item[-] $h$ is the reference of $\block$, i.e., a hash of the previous fields.
    \item[-] $\mbf{F}$ is the fruit set to be included in $\block$.
\end{itemize}
\begin{definition}[Block validity]\label{def:valid_block}
A block $\block=\langle\langle \previous,h_f,\nonce,\msf{dig},\record,$ $h\rangle,\mbf{F}\rangle$ is \emph{valid}, if the following hold:
\begin{enumerate}
\item $\msf{dig}=\digest(\mbf{F})$.
\item $F$ is a valid fruit set.
    \item $\oracle(\previous||h_f||\nonce||\msf{dig}||\record)=h$.
    \item $[h]_{:\secpar}<D_{\pb}$, where $[h]_{:\secpar}$ denotes the first $\secpar$ bits of $h$, and $D_{\pb}$ is the difficulty value such that the probability that an input satisfies the relation is less than $\pb$. 
\end{enumerate}
\end{definition}

\begin{remark}
By making a query to $\oracle(\cdot)$, the party cannot know in advance whether the response hash value $h$ will lead to the successful mining of a fruit and/or block, or neither of two. Thus, the fields $\previous$ and $\msf{dig}$ are included in a fruit (resp. $h_f$ and $\record$ are included in a block) only for mining purposes. 
\end{remark}

Let $\chain_i$ be the ledger state in the view of party $\party{i}$. By $\chain_i[j]:=\langle\langle \previous^{i,j},$ $h_f^{i,j},\nonce^{i,j},\msf{dig}^{i,j},$ $\record^{i,j},h^{i,j}\rangle,\mbf{F}^{i,j}\rangle$, we denote the $j$-th block of $\chain_i$ and by $|\chain_i|$ the length of $\chain_i$. 

\begin{definition}[Fruit recency]\label{def:fruit_recency}
A fruit  $\fruit:=\langle \previous,h_f,\nonce,\msf{dig},\record,h\rangle$ is \emph{recent} w.r.t. $\chain_i$ if it points to some of the last $\recency\cdot\kappa$ blocks of $\chain_i$, i.e., there exists some $k>|\chain_i|-\recency\cdot\kappa$ such that $h_f=h^{i,k}$. 
\end{definition}

\begin{definition}[Chain validity]\label{def:valid_chain}
A chain $\chain_i$ is \emph{valid}, if the following hold:
\begin{enumerate}
    \item The chain is rooted at the special ``genesis'' block, i.e., 
    \[\chain_i[0]=\langle\langle0,0,0,0,\bot,\oracle(0,0,0,0,\bot)\rangle,\emptyset\rangle\;.\]
    \item Each block is valid according to Definition~\ref{def:valid_block} and refers to the previous block's reference, i.e., 
    \[\forall j\in[|\chain_i|]: \previous^{i,j}=h^{i,j-1}\;. \] 
    \item For every $j\in[|\chain_i|]$ and every $\fruit:=\langle \previous,h_f,\nonce,\msf{dig},\record,$ $h\rangle\in\mbf{F}^{i,j}$, there exists some $k>j-\recency\cdot\kappa$ such that $h_f=h^{i,k}$.
\end{enumerate}
\end{definition}

For completeness, we describe the FruitChain PoW protocol $\protocol$ as presented in~\cite{fruitchain}, parameterized by $\oracle(\cdot),\digest(\cdot),$ $\pb,\pf,\recency$, in Figure~\ref{fig:fruitchain_description}.

\subsection{The FruitChain Protocol in the EVP Framework}\label{subsec:fruitchain_framework}
We abstract the FruitChain protocol $\protocol$ as the protocol $\protocolevp$, that specifies the oracles below:
\begin{enumerate}
    \item The \emph{longest chain oracle} $\mc{O}_\msf{lc}$: receives as input a state $\chain$ and an array of blocks $\tilde{\mathbb{B}}$. It stores $\tilde{\mathbb{B}}$ in its memory.
    Given $\chain$ and all the arrays of blocks that are stored in its memory, it constructs a set $\mbf{A}$ that includes all the chains that can be formed. It checks which of these chains are valid according to Definition \ref{def:valid_chain} and constructs a set $\mbf{A}'\subset \mbf{A}$ with these chains. 
    It finds the longest chain(s) of $\mbf{A}'$ denoted by $\chain^1, \ldots, \chain^l$. If $l> 1$, then it finds $i_0 \in \lbrace 1,\ldots,l\rbrace$ such that the last block in $\chain^{i_0}$, denoted by $\chain^{i_0} [|\chain^{i_0}|]$, appears first in $\tilde{\mathbb{B}}$.
    If $|\chain^{i_0}|>|chain|$, it sets $\chain' \leftarrow \chain^{i_0}$, otherwise $\chain' \leftarrow \chain$. It sets as $h_f$ the reference of $\chain'[\mathsf{max}\{1,$ $|\chain'|-\kappa\}]$, and as $\previous$ the reference of $\chain'[|\chain'|-1]$. 
    It extracts the sequence of records
    $\big(\record_{1},\ldots,\record_{\ell_2}\big)$ included in the fruits of $\chain'$ (by executing the procedure $\mathtt{Extract\_Fruit}(\chain')$ in Figure~\ref{fig:fruitchain_description} without checking again if $\chain'$ is valid).
    It outputs $\chain',\previous,h_f$ and $\big(\record_{1},\ldots,\record_{\ell_2}\big)$. The party can make up to $1$ query per round and the single query cost is $C_\msf{lc}$.

    %
    \item The \emph{fruit set oracle} $\mc{O}_\msf{fs}$: receives as input a state $\chain'$ and two sets of fruits $\tilde{\mbf{F}},\tilde{\mbf{F'}}$. It finds the subset $\mbf{X}$ of $\tilde{\mbf{F'}}$ that includes all the valid fruits of $\tilde{\mbf{F'}}$ according to Definition~\ref{def:valid_fruit}. It returns the set of fruits $\mbf{F}=\tilde{\mbf{F}} \cup \mbf{X} $, the set $\mbf{F}_\msf{rec}\subseteq\mbf{F}$ of valid fruits that are recent w.r.t. $\chain'$ for recency parameter $r$ (Definition~\ref{def:fruit_recency}) and are not already in $\chain'$, and the digest $\digest(\mbf{F}_\msf{rec})$. 
    The party can make up to $1$ query per round and the single query cost is $C_\msf{fs}$.
    \item The \emph{transaction oracle} $\mc{O}_\msf{tx}$: receives as input a set of transactions $\{\tx_1,\ldots,\tx_{\ell_1}\}$ and a sequence of records $(\record_1,\ldots,\record_{\ell_2})$. It computes a record of transactions $\record$ that includes all the transactions of $\{\tx_1,\ldots,\tx_{\ell_1}\}$ that are valid according to $(\record_1,\ldots,\record_{\ell_2})$, where transaction validity is defined in a protocol-specific manner. It outputs $\record$. The party can make up to $1$ query per round and the single query cost is $C_\msf{tx}$.
    \item The \emph{random oracle} $\mc{O}_\msf{ro}$:  on a general query $x\in\{0,1\}^*$ checks if there is a stored pair $(x,\cdot)$. If there is not such a pair, it randomly samples an image denoted by $H(x)$ from $\{0,1\}^{2\secpar}$, stores $(x,H(x))$ and returns $H(x)$. Otherwise, it returns $(x,H(x))$.
    In our setting, the queries will have the form $\previous||h_f||\nonce||\digest(\mbf{F}_\msf{rec})||\record$ and the response will be a reference~$h$. The party can make up to $q$ queries per round and the single query cost is $C_\msf{ro}$.

\end{enumerate}
\begin{figure}[H]
\centering
\begin{thickframe}
\begin{mdframed}
{
\underline{\emph{The protocol $\protocol$.}}\\[5pt]
\begin{small}
\noindent The protocol is initialized as $\chain_i=\langle\langle0,0,0,0,\bot,$ $\oracle(0,0,0,0,\bot)\rangle,\emptyset\rangle$ and $\mbf{F}^i=\emptyset$. Let $\mbf{F}^i_\msf{rec}$ be the set of fruits that are recent w.r.t. $\chain_i$ and are not already in $\chain_i$ (initialized as empty). Let $h_f$ be the reference of $\chain_i[\mathsf{max}\{1,|\chain_i|-\kappa\}]$. Let $\previous$ be the reference of $\chain_i[|\chain_i|-1]$.\\[3pt]
$\blacktriangleright$ Upon receiving a fruit $\fruit$, if $\fruit$ is valid, then set $\mbf{F}^i\leftarrow\mbf{F}^i\cup\fruit$.\\[2pt]
$\blacktriangleright$ Upon receiving a state $\chain'_i$, if $\chain'_i$ is valid and $|\chain'_i|>|\chain_i|$, then update $\chain_i$ as $\chain_i\leftarrow\chain'_i$.\\[2pt]
$\blacktriangleright$ Upon receiving an input $\record$ from $\mc{Z}$, 
\begin{enumerate}
    \item Pick a random nonce $\nonce\in\{0,1\}^\kappa$.
    \item Compute $h\leftarrow\oracle(\previous||h_f||\nonce||\digest(\mbf{F}^i_\msf{rec})||\record)$.
    \item If $[h]_{-\secpar}<D_{\pf}$ (fruit mining successful), then 
    \begin{enumerate}
        \item Set $\fruit\leftarrow\langle\previous,h_f,\nonce,\digest(\mbf{F}^i_\msf{rec}),\record,h\rangle$.
        \item Set $\mbf{F}^i\leftarrow\mbf{F}^i\cup\fruit$.
        \item Broadcast $\fruit$.
    \end{enumerate}
    \item If $[h]_{:\secpar}<D_{\pb}$ (block mining successful), then
    \begin{enumerate}
        \item Set $\block\leftarrow\langle\langle\previous,h_f,\nonce,\digest(\mbf{F}^i_\msf{rec}),\record,h\rangle, \mbf{F}^i_\msf{rec}\rangle$.
        \item Update $\chain_i$ as $\chain_i\leftarrow\chain_i||\block$.
        \item Broadcast $\chain_i$.
    \end{enumerate}
    \item Return $\mathtt{Extract\_Fruit}(\chain_i)$ to $\mc{Z}$, where the procedure $\mathtt{Extract\_Fruit}(\cdot)$ is described below.
    \item[]
\end{enumerate}
\end{small}
}
\end{mdframed}
\begin{mdframed}
{
\underline{\emph{The procedure $\mathtt{Extract\_Fruit}(\chain_i)$.}}\\[3pt]
\begin{small}
If $\chain_i$ is valid, then
\begin{enumerate}
    \item Extract a sequence of distinct fruits $(\fruit_1,\ldots,\fruit_\ell)$ from $\chain_i$ by including only the first occurrence in case a fruit $\fruit_k$ is included multiple times.
    \item Order the extracted fruits $\fruit_1,\ldots,\fruit_\ell$ according to the first block that contains the fruit. For fruits on the same block, follow the order in which the fruits are serialized within the block. Let $\fruit_{\sigma(1)},\ldots,\fruit_{\sigma(\ell)}$ be the ordering of the fruit sequence according to permutation $\sigma(\cdot)$.
    \item Output the sequence of records $\big(\record_{\sigma(1)},\ldots,\record_{\sigma(\ell)}\big)$, where $\record_{\sigma(k)}$ is the record contained in fruit $\fruit_{\sigma(k)}$.
\end{enumerate} 
\end{small}
}
\end{mdframed}
\end{thickframe}
\caption{The protocol $\protocol$ for party $\party{i}$.}
\label{fig:fruitchain_description}
\end{figure}
\normalsize

\begin{figure}[H]
\centering
\begin{thickframe}
\begin{mdframed}
{
\underline{\emph{The protocol $\protocolevp$.}}\\[5pt]
\begin{small}
\noindent The protocol is initialized as $\chain_i=\langle\langle0,0,0,0,\bot,$ $\oracle(0,0,0,0,\bot)\rangle,\emptyset\rangle$ and $\mbf{F}^i=\emptyset$. In addition, a flag $\msf{success}_i$ is initialized as $0$.\\[3pt]
$\blacktriangleright$ In each round $T$, upon receiving a set of transactions $\{\tx_1,\ldots,\tx_{\ell_1^T}\}$ from $\mc{Z}$, execute the following steps:
\begin{enumerate}
    \item Retrieve the set of all fruits and the array of all blocks that were diffused during the previous round $T-1$, denoted by $\tilde{\mbf{F}}^i_{T-1}$ and $\tilde{\mathbb{B}}^i_{T-1}$, respectively. Note that $\tilde{\mathbb{B}}^i_{T-1}$ is such that the block received first in $T-1$ is in the first position.

    \item Make the query $(\chain_i,\tilde{\mathbb{B}}^i_{T-1})$ to $\mc{O}_\msf{lc}$ and receive $\chain'_i,\previous,h_f$, and $(\record_1,\ldots,\record_{\ell_2^T})$. 
    \item Set $\chain_i\leftarrow\chain'_i$.
     \item Make the query $(\chain_i,\mbf{F}^i,\tilde{\mbf{F}}^i_{T-1})$ to $\mc{O}_\msf{fs}$ and receive $\tilde{\mbf{F}}^i,\mbf{F}^i_\msf{rec}$, and $\digest(\mbf{F}^i_\msf{rec})$. 
    \item Set $\mbf{F}^i\leftarrow\tilde{\mbf{F}}^i$. 
    \item Make the query $\big(\{\tx_1,\ldots,\tx_{\ell_1^T}\},(\record_1,\ldots,\record_{\ell_2^T})\big)$ to $\mc{O}_\msf{tx}$ and receive $\record$.
    \item For $k=1,\ldots,q$:
    \begin{enumerate}
     \item Pick a random nonce $\nonce_k\in\{0,1\}^\kappa$.
    \item Make the query $\previous||h_f||\nonce_k||\digest(\mbf{F}^i_\msf{rec})||\record$ to $\mc{O}_\msf{ro}$ and receive $h_k$.
    \item If $[h_k]_{-\secpar}<D_{\pf}$, then 
    \begin{enumerate}
        \item Set $\fruit\leftarrow\langle\previous,h_f,\nonce_k,\digest(\mbf{F}^i_\msf{rec}),\record,h_k\rangle$.
        \item Set $\mbf{F}^i\leftarrow\mbf{F}^i\cup\fruit$.
        \item Send $\fruit$ to the Diffuse functionality.
    \end{enumerate}
    \item If $[h_k]_{:\secpar}<D_{\pb}$ and $\msf{success}_i=0$, then
    \begin{enumerate}
        \item Set $\block\leftarrow\langle\langle\previous,h_f,\nonce_k,\digest(\mbf{F}^i_\msf{rec}),\record,h_k\rangle, \mbf{F}^i_\msf{rec}\rangle$.
        \item Set $\chain_i\leftarrow\chain_i||\block$.
        \item Set $\msf{success}_i\leftarrow1$.
        \item Send $\block$ to the Diffuse functionality.
    \end{enumerate}
     \end{enumerate}
    \item Send $\tilde{\mathbb{B}}^i_{T-1}$ to the Diffuse functionality.
    \item Set $\msf{success}_i\leftarrow0$.
    \item Return a $(\msf{complete},T)$ message to $\mc{Z}$.
   
\end{enumerate}
\end{small}
}
\end{mdframed}
\end{thickframe}
\caption{The protocol $\protocolevp$ for party $\party{i}$.}
\label{fig:fruitchain_evp}
\end{figure}
\normalsize

Given the description of $\mc{O}_\msf{lc},\mc{O}_\msf{fs},\mc{O}_\msf{tx},\mc{O}_\msf{ro}$, and the terminology in Subsection~\ref{subsec:basic}, the blockchain protocol $\protocolevp$ is presented in Figure~\ref{fig:fruitchain_evp}. At this point, we provide an overview of the  $\protocolevp$ protocol.\\[2pt]
\indent\emph{Overview of  $\protocolevp$.}
Each party stores all the fruits that are valid. Note that the validity of each fruit (cf. Definition \ref{def:valid_fruit}) does not depend on which chain constitutes the ledger, unlike the recency of the fruit (cf. Definition \ref{def:fruit_recency}). 
\par 
During each round, when a party is activated, it receives a set of transactions $\{\tx_1,\ldots,\tx_{\ell_1^T}\}$ as input from the environment, and retrieves all the fruits and blocks diffused in the previous round. Then, it gives as input to the longest chain oracle $\mc{O}_\msf{lc}$ its current chain and the blocks it retrieved, and it receives as output (i) the chain $\chain'_i$ that the party will extend and is the longest valid chain, (ii) the hash of the last block of this chain, (iii) the hash of the block to which the fruits that will be produced in this round will point and (iv) all the records of $\chain'_i$.
\par Next, the party gives as input to the fruit set oracle $\mc{O}_\msf{fs}$ the chain $\chain'_i$, the set of the valid fruits it retains and the fruits it retrieved. It receives as output (i) the updated set of the valid fruits that includes also the fruits that it retrieved and were valid (ii) the valid fruits that are recent w.r.t. $\chain'_i$  and are not already in  $\chain'_i$, and (iii) the digest of the set of these fruits which works as a ``fingerprint''.
\par  Afterwards, the party makes a query to the transaction oracle $\mc{O}_\msf{tx}$ with input the transactions $\{\tx_1,\ldots,\tx_{\ell_1^T}\}$ it received from the environment and the records $(\record_1,\ldots,\record_{\ell_2^T})$ that received from $\mc{O}_\msf{lc}$. It outputs a record that includes the transactions of $\{\tx_1,\ldots,\tx_{\ell_1^T}\}$ that were valid w.r.t. $(\record_1,\ldots,\record_{\ell_2^T})$ (transaction validity is defined in a protocol-specific way) and will be included in the instance that will be used for the queries to the random oracle $\mc{O}_\msf{ro}$.
\par Finally, it makes $q$ queries to the random oracle with input an \emph{instance} that includes: (a) the hash of the last block in $\chain'_i$, (b) the hash of the block to which the new fruits will point, (c) a nonce, (d) the digest of the recent fruits, and (e) the output of the transaction oracle. When it receives an output from the random oracle, it checks (i) if the last $\secpar$ bits are lower than $D_{\pf}$, and (ii) if the first $\secpar$ bits are lower than $D_{\pb}$. If (i) holds, a fruit has been produced, so it sends this fruit to the Diffuse functionality. If (ii) holds, then a block has been produced, so it sends this block to the Diffuse functionality and stops checking (ii) in the queries.\\[2pt] 
\indent\emph{Assignment of rewards.}
We consider that every fruit included in a block that is part of the ledger mints  $\reward$ rewards and sends them to the party specified in the `` coinbase '' transaction\footnote{\url{https://en.bitcoin.it/wiki/Coinbase}}  in the fruit's record $\record$. Note that the rewards and costs of querying the protocol's oracles are in the same unit.

\noindent Below, we discuss the differences between the descriptions of $\protocolevp$ and $\protocol$.

\begin{enumerate}
    \item The environment provides transactions, not records. In addition, $\protocolevp$ checks if the transactions provided by the environment are valid according to the party's local chain. We do not specify when a transaction is valid according to the party's local chain because this depends on the format of the transactions the protocol accepts.  
    \item (i) The party makes $q$ queries to the random oracle $\mc{O}_\msf{ro}$ per round (instead of $1$) (ii) each party can produce at most one block per round. These modifications are in line with the execution model of \cite{backbone}.
    \item The party returns just `$\msf{complete}$'  and the number of the round to the environment, instead of a sequence of records. Note that as the party diffuses its blocks, the environment can receive them via its interaction with the adversary. 
    \item The party diffuses only the blocks that it received during the previous round and the blocks and fruits that it has produced during the current round; unlike in $\protocol$, the party does not diffuse its whole chain every time it produces a new block. On the other hand, the longest chain oracle $\mc{O}_\msf{lc}$ stores all the blocks it receives from the beginning of the execution. This approach reflects the realistic setting where the parties' local chains are not communicated over the network during the mining process.
    Instead, only the newly mined blocks are normally diffused, and the miners can reconstruct all possible chains given the received blocks they have recorded throughout the execution \footnote{\url{https://wiki.bitcoinsv.io/index.php/Main_Page}}. Note that when the network is synchronous, which means that at the end of each round all the messages diffused by honest parties are delivered to every other honest party, the approach in $\protocolevp$ ``implies'' the one in $\protocol$. Namely, due to synchronicity,
    in the beginning of round $T$, every honest party $
    \party{i}$ can \emph{recursively} reconstruct the local state $\chain_j$ of another honest party $
    \party{j}$  given its view of $
\chain_j$ in the beginning of the previous round $T-1$ and the blocks of $
    \party{j}$ that $
    \party{i}$ received by the end of $T-1$ (note that all honest parties' states are initialized as $\langle\langle0,0,0,0,\bot,\oracle(0,0,0,0,\bot)\rangle,\emptyset\rangle$, so recursive reconstruction is feasible across honest parties as rounds progress).
    Besides, at any moment during round $T$, $\party{i}$ can reconstruct $\chain_j$ via its view of $\chain_j$ in the beginning of $T$ and the blocks and fruits received from $
    \party{j}$ since the beginning of the round.      
    \item The recipient of the fruit's rewards is the party specified in the ``coinbase'' transaction. In \cite{fruitchain} the fruit's rewards are shared evenly among the miners of the fruits that belong to a preceding part of the ledger. Note that we follow the approach of \cite{EVP} which is equivalent to the approach of \cite{fruitchain} when we do not take into account transaction fees, which means that we assume that each fruit gives the same rewards.
    
\end{enumerate}

\section{The Single Pool Protocol}\label{sec:single}
In this section, we provide the definition of a pool in a blockchain system and we describe the rules of a single pool in \textsc{FruitChain}, denoted by $\single$, that includes all the parties. In this pool, all the parties ask the random oracle, but only the pool leader asks the longest chain, the fruit set and the transaction oracle, and determines the instance that will be used by all the parties for the queries to the random oracle.  In the next section, we will prove that joining this ``centralised'' pool is an EVP.  
    
\subsection{Definition of a Pool}\label{subsec:definition_pool}
Intuitively, a pool of some protocol $\Pi$ comprises a subset of parties in $\Pi$ that collaborate by interacting internally according to some well-specified communication pattern and guidelines. Formally, we provide the following definition.

\begin{definition}[Pool]\label{def:pool}
Let $\Pi$ be a blockchain protocol with parties $\party{1},\ldots,\party{n}$. A \emph{pool of $\Pi$} is a quadruple $\langle\mbf{V},\mbf{E},\mc{F}_\msf{comm},\tilde{\Pi}\rangle$, where
\begin{itemize}
    \item $\mbf{V}\subseteq\{\party{1},\ldots,\party{n}\}$ is a subset of parties in $\Pi$.
    \item $\mbf{E}\subseteq\{(\party{i},\party{j})|\party{i},\party{j}\in\mbf{V}\}$ is a subset of pairs of parties in $\mbf{V}$ that determines the available simplex communication connections among parties in $\mbf{V}$.
    \item $\mc{F}_\msf{comm}$ is a communication functionality that supports the parties' interaction w.r.t. $\mbf{E}$.
    \item $\tilde{\Pi}$ is a protocol executed by parties in $\mbf{V}$ that captures the execution instructions for the parties in $\mbf{V}$. In addition, $\tilde{\Pi}$ allows parties to have access to the Diffuse functionality (cf. Subsection~\ref{subsec:framework_execution}), hence to the messages exchanged during the execution of $\Pi$.%
\end{itemize}
\end{definition}

\subsection{A Single Pool of $\protocolevp$}\label{subsec:single_pool}
Given Definition~\ref{def:pool}, we specify a \emph{single pool of $\protocolevp$ with leader $\party{L}$} as the quadruple $\langle\mbf{V},\mbf{E},\mc{F}_\msf{auth}(\mbf{E}),\single\rangle$ where
\begin{itemize}
    \item $\mbf{V}:=\{\party{1},\ldots,\party{n}\}$, i.e., all the parties collaborate. For some $i^*\in[n]$, we have that $\party{L}=\party{{i^*}}$.
    \item $\mbf{E}:=\{(\party{L},\party{i}),(\party{i},\party{L})\}_{i\in[n]\setminus\{i^*\}}$. Namely, the pool leader $\party{L}$ can communicate with every other party and vice versa. Note that the non leader parties do not communicate with each other.
    \item $\mc{F}_\msf{auth}(\mbf{E})$ is the \emph{message authentication functionality} w.r.t. $\mbf{E}$, defined in the spirit of~\cite{Canetti04} as follows:
\begin{itemize}
    \item[$\blacktriangleright$] Upon receiving $(\textsc{Send},\party{j},M)$ from $\party{i}$, if $(\party{i},\party{j})\in\mbf{E}$, then $\mc{F}_\msf{auth}(\mbf{E})$ sends the message $(\textsc{Sent},\party{i},M)$ to $\party{j}$.
\end{itemize}
Similar to~\cite{Canetti04}, $\mc{F}_\msf{auth}(\mbf{E})$ can be implemented via digital signatures and some setup assumption, such as the presence of a certification authority or the out-of-band exchange of verification keys among the parties in the pool. 
   \item $\single$ is executed by $\party{1},\ldots,\party{n}$ and defines each party's deviation from the protocol $\protocolevp$. During the execution of $\single$, $\party{L}$ takes over the cost for setting up an instance to the random oracle in each round. Then, all parties contribute to the fruit and block mining effort w.r.t. this instance. Upon successful mining of a block, $\party{L}$ shares the rewards that correspond to the fruits included in this block according to the guidelines. The protocol $\single$ is formally introduced in the following subsection.
\end{itemize}

\subsection{Protocol Description}\label{subsec:single_description}
First, we describe an additional oracle that $\single$ utilizes and provide its overview. 
\par 
The \emph{light transaction verification oracle}  $\mc{O}_\msf{ltx}$: receives as input a record of transactions $\record$ and a transaction $\tx$. It outputs $1$ if the transaction $\tx$ is included in the record $\record$ and $\tx$ is valid \footnote{The transaction validity is defined in a protocol-specific manner.}, and $0$ otherwise. The party can make up to $1$ query per round and the cost of a single query is $C_\msf{ltx}$\footnote{This oracle reflects a procedure similar to the ``simplified payment verification'' \url{https://wiki.bitcoinsv.io/index.php/Simplified_Payment_Verification}. $C_\msf{ltx}$ is significantly lower compared to  $C_\msf{tx}$. }.\\[2pt]
\indent\emph{Overview of $\single$.}
During each round, the pool leader asks $\mc{O}_\msf{lc},\mc{O}_\msf{fs}$, and $\mc{O}_\msf{tx}$, and creates the instance that will be used for the queries to the random oracle. Then, it sends this instance to the pool members. The pool leader and the other members ask the random oracle $q$ queries when they are activated. When a fruit or a block is produced, they send it to the Diffuse functionality (at most one block per round). 
\par Both the pool leader and the other pool members count the cost that the pool leader should incur for the oracles $\mc{O}_\msf{lc},\mc{O}_\msf{fs},\mc{O}_\msf{tx}$. When a block that uses the instance sent by the pool leader has been diffused, the pool leader creates a payment transaction in the next round. The payments are as follows: if the cost that the pool leader incurred for creating the instances since the last block is higher than the block's rewards (which is equal to the number of fruits multiplied by the fruit reward $\rf$), then the pool leader holds all the rewards. If the block's rewards are higher, then the pool leader subtracts the cost and shares the remaining rewards equally among all the members of the pool including itself. Note that the members will check if the payments have been computed correctly via the light transaction verification oracle \footnote{A similar countermeasure has also been used in P2pool
\url{https://bitcoinmagazine.com/technical/p2pool-bitcoin-mining-decentralization}}.
In addition, to prevent block withholding attacks (cf. \cite{ Rosenfeld2011AnalysisOB,7163020}), both the pool leader and the other pool members check if the diffused fruits and blocks use the instance sent by the pool leader; if not, they abandon the pool.
The $\single$ protocol is presented in  Figures~\ref{fig:single_leader} and~\ref{fig:single_other}.

\begin{figure}[H]

\begin{thickframe}
\begin{mdframed}
{
\underline{\emph{The protocol $\single$ for pool leader $\party{L}$.}}\\[5pt]
\begin{small}
\noindent The protocol is initialized as $\chain=\langle\langle0,0,0,0,\bot,$ $\oracle(0,0,0,0,\bot)\rangle,\emptyset\rangle$, a list $record$ as empty and $\mbf{F}=\emptyset$. In addition, a flag $success$ is initialized as $0$, a quadruple of variables $(inst_1,inst_2,inst_3,inst_4)$ as $(0,0,0,\bot)$, and variables $cost,rewards,W_L,W$ as $0$.\\[3pt]
$\blacktriangleright$ In each round $T$, upon receiving a set of transactions $\{\tx_1,\ldots,\tx_{\ell_1^T}\}$ from $\mc{Z}$, execute the following steps:
\begin{enumerate}
    \item Retrieve the set of all fruits and the array of all blocks that were diffused during the previous round $T-1$, denoted by $\tilde{\mbf{F}}_{T-1}$ and $\tilde{\mathbb{B}}_{T-1}$, respectively. Note that $\tilde{\mathbb{B}}_{T-1}$ is such that the block received first in $T-1$ is in the first position.
   \item If $\tilde{\mathbb{B}}_{T-1}$ contains a block $\hat{\block}:=\langle\langle \hat{h}_{-1},\hat{h}_f,\hat{\nonce},\hat{\msf{dig}},\hat{\record},\hat{h}\rangle,\hat{\mbf{F}}\rangle$ or $\tilde{\mbf{F}}_{T-1}$ contains a fruit $\langle \hat{h}_{-1},\hat{h}_f,\hat{\nonce},\hat{\msf{dig}},\hat{\record},\hat{h}\rangle$ such that $(\hat{\previous},\hat{h}_f,\hat{\msf{dig}},\hat{\record})\neq(inst_1,inst_2,inst_3,inst_4)$, then dissolve the pool and from this round and on proceed by following the fallback $\protocolevp$ (cf. Figure~\ref{fig:fruitchain_evp}).
   \item If $\tilde{\mathbb{B}}_{T-1}$ is not empty, then
   \begin{enumerate}
       \item Parse the first block in $\tilde{\mathbb{B}}_{T-1}$ as $\langle\langle \tilde{h}_{-1},\tilde{h}_f,\tilde{\nonce},\tilde{\msf{dig}},\tilde{\record},\tilde{h}\rangle,\tilde{\mbf{F}}\rangle$ and set the total rewards as $rew\leftarrow|\tilde{\mbf{F}}|\cdot\reward$.
       \item Set the leader's payment as 
       \[W_L\leftarrow \mathrm{min}\{cost,rew\}+\mathrm{max}\Big\{\dfrac{rew-cost}{n},0\Big\}.\]
       \item Set the payment of every other party as
       \[W\leftarrow\mathrm{max}\Big\{\dfrac{rew-cost}{n},0\Big\}.\]
       \item Create the special transaction $\tx_T$ that includes the payment of every party $\party{i}\neq\party{L}$ for round $T$.
       \item Make the query $(\chain,\tilde{\mathbb{B}}_{T-1})$ to $\mc{O}_\msf{lc}$ and receive $\chain',\previous,h_f$, and $(\record_1,\ldots,\record_{\ell_2^T})$. 
       \item Set $\chain\leftarrow\chain'$, $inst_1\leftarrow\previous$, $inst_2\leftarrow h_f$, and $record\leftarrow(\record_1,\ldots,\record_{\ell_2^T})$.
       \item Set $cost\leftarrow\cost_\msf{lc}$.
   \end{enumerate}
     \item Make the query $(\chain,\mbf{F},\tilde{\mbf{F}}_{T-1})$ to $\mc{O}_\msf{fs}$ and receive $\mbf{F}',\mbf{F}_\msf{rec}$, and $\digest(\mbf{F}_\msf{rec})$.
    \item Set $\mbf{F}\leftarrow\mbf{F}'$ and $inst_3\leftarrow\digest(\mbf{F}_\msf{rec})$.
    \item Set $cost\leftarrow cost+\cost_\msf{fs}$.
    \item If $\tilde{\mathbb{B}}_{T-1}$ is not empty, then make the query $\big(\{\tx_1,\ldots,\tx_{\ell_1^T}\}\cup\{\tx_T\},record\big)$ to $\mc{O}_\msf{tx}$ and receive $\record$. Else, make the query $\big(\{\tx_1,\ldots,\tx_{\ell_1^T}\},record\big)$ to $\mc{O}_\msf{tx}$ and receive $\record$.
    \item Set $inst_4\leftarrow\record$.
    \item Set $cost\leftarrow cost+\cost_\msf{tx}$.
    \item For every party $\party{i}\neq\party{L}$, if $\tilde{\mathbb{B}}_{T-1}$ is not empty send $(\textsc{Send},\party{i},(inst_1,inst_2, inst_3,inst_4,$ $\mbf{F}_\msf{rec},T,\tx_T))$ to $\mc{F}_\msf{auth}(\mbf{E})$, else send $(\textsc{Send},\party{i},$ $(inst_1,inst_2,inst_3,inst_4,\mbf{F}_\msf{rec},T))$ to $\mc{F}_\msf{auth}(\mbf{E})$.
    \item Participate in the fruit and block mining process by querying $\mc{O}_\msf{ro}$ like any party (cf. Figure~\ref{fig:single_other}). 
    \item Return a $(\msf{complete},T)$ message to $\mc{Z}$.
   
\end{enumerate}
\end{small}
}
\end{mdframed}
\end{thickframe}
         \caption{The protocol $\single$ for the leader $\party{L}$.}
         \label{fig:single_leader}
     \end{figure}
\normalsize
\begin{figure}[H]

\begin{thickframe}
\begin{mdframed}
{
\underline{\emph{The protocol $\single$ for non leader $\party{i}$.}}\\[5pt]
\begin{small}
\noindent A flag $success_i$ is initialized as $0$ and a variable $cost$ is initialized as $0$. A quadruple of variables $(inst_1,inst_2,inst_3,inst_4)$ is initialized as $(0,0,0,\bot)$.\\[3pt]
$\blacktriangleright$ In each round $T$, upon receiving  a message  $(\textsc{Sent},\party{L},(\previous,h_f,\digest(\mbf{F}_\msf{rec}),\record,\mbf{F}_\msf{rec},T,\tx_T))$ or $(\textsc{Sent},\party{L},(\previous,h_f,\digest(\mbf{F}_\msf{rec}),\record,\mbf{F}_\msf{rec},T))$ from $\mc{F}_\msf{auth}(\mbf{E})$, store the tuple $(\previous,h_f,\digest(\mbf{F}_\msf{rec}),\mbf{F}_\msf{rec},T,\tx_T)$ or $(\previous,h_f,\digest(\mbf{F}_\msf{rec}),\mbf{F}_\msf{rec},T)$ respectively.
 
$\blacktriangleright$ In each round $T$, upon receiving a set of transactions $\{\tx_1,\ldots,\tx_{\ell_i^T}\}$ from $\mc{Z}$, execute the following steps:
\begin{enumerate}
    \item Retrieve the set of all fruits and the array of all blocks that were diffused during the previous round $T-1$, denoted by $\tilde{\mbf{F}}^i_{T-1}$ and $\tilde{\mathbb{B}}^i_{T-1}$ respectively. Note that $\tilde{\mathbb{B}}_{T-1}$ is such that the block received first in $T-1$ is in the first position.

   \item If any of the following hold, then leave the pool and from this round and on proceed by following the fallback $\protocolevp$ (cf. Figure~\ref{fig:fruitchain_evp}):
   \begin{enumerate}
       \item $\tilde{\mathbb{B}}^i_{T-1}$ contains a block $\hat{\block}:=\langle\langle \hat{h}_{-1},\hat{h}_f,\hat{\nonce},$ $\hat{\msf{dig}},\hat{\record},\hat{h}\rangle,\hat{\mbf{F}}\rangle$ or $\tilde{\mbf{F}}^i_{T-1}$ contains a fruit $\langle \hat{h}_{-1},\hat{h}_f,\hat{\nonce},\hat{\msf{dig}},\hat{\record},\hat{h}\rangle$ such that $(\hat{\previous},\hat{h}_f,\hat{\msf{dig}},\hat{\record})\neq(inst_1,inst_2,inst_3,inst_4)$.
       \item $\tilde{\mathbb{B}}_{T-1}$ is not empty, $\langle\langle \tilde{h}_{-1},\tilde{h}_f,\tilde{\nonce},\tilde{\msf{dig}},\tilde{\record},\tilde{h}\rangle,\tilde{\mbf{F}}\rangle$ is the first block in $\tilde{\mathbb{B}}_{T-1}$, and one of the following happens: (i) $\party{i}$ received $(\textsc{Sent},\party{L},(\previous,h_f,\digest(\mbf{F}_\msf{rec}),\record,\mbf{F}_\msf{rec},T))$ (ii) $\tx_T$ does not include  $\party{i}$ as recipient or the amount that is sent to $\party{i}$ differs from  $\mathrm{max}\Big\{\dfrac{|\tilde{\mbf{F}}|\cdot\reward-cost}{n},0\Big \}$ (iii) it makes a query to the oracle $\mc{O}_\msf{ltx}$ with input $\tilde{\record}$ and  $\tx_T$ and it gets output 0.
       \item $\party{i}$ did not receive a message  $(\textsc{Sent},\party{L},(\previous,h_f,\digest(\mbf{F}_\msf{rec}),\record,\mbf{F}_\msf{rec},T))$ or $(\textsc{Sent},\party{L},(\previous,h_f,\digest(\mbf{F}_\msf{rec}),\record,\mbf{F}_\msf{rec},T,\tx_T))$ from $\mc{F}_\msf{auth}(\mbf{E})$.
   \end{enumerate}
    \item If $\tilde{\mathbb{B}}_{T-1}$ is not empty then $cost \leftarrow C_\msf{lc}+ C_\msf{fs}+C_\msf{tx}$ otherwise $cost \leftarrow cost+ C_\msf{fs}+C_\msf{tx}$.
    \item Read the values $\previous,h_f,\digest(\mbf{F}_\msf{rec}),\mbf{F}_\msf{rec},T$ and set $inst_1\leftarrow\previous$, $inst_2\leftarrow h_f$, $inst_3\leftarrow \digest(\mbf{F}_\msf{rec})$, and $inst_4\leftarrow\record$.

    \item For $k=1,\ldots,q$:
    \begin{enumerate}
     \item Pick a random nonce $\nonce_k\in\{0,1\}^\kappa$.
    \item Make the query $inst_1||inst_2||\nonce_k||inst_3||inst_4$ to $\mc{O}_\msf{ro}$ and receive $h_k$.
    \item If $[h_k]_{-\secpar}<D_{\pf}$, then 
    \begin{enumerate}
        \item Set $\fruit\leftarrow\langle inst_1,inst_2,\nonce_k,inst_3,inst_4,h_k\rangle$.
        \item Send $\fruit$ to the Diffuse functionality.
    \end{enumerate}
    \item If $[h_k]_{:\secpar}<D_{\pb}$ and $\msf{success}_i=0$, then
    \begin{enumerate}
        \item Set $\block\leftarrow\langle\langle inst_1,inst_2,\nonce_k,inst_3,inst_4,h_k\rangle, \mbf{F}_\msf{rec}\rangle$.
        \item Set $\msf{success}_i\leftarrow1$.
        \item Send $\block$ to the Diffuse functionality.
    \end{enumerate}
     \end{enumerate}
    \item Send $\tilde{\mathbb{B}}^i_{T-1}$ to the Diffuse functionality.
    \item Delete the tuple $(\previous,h_f,\digest(\mbf{F}_\msf{rec}),\mbf{F}_\msf{rec},T)$ or $(\previous,h_f,\digest(\mbf{F}_\msf{rec}),\mbf{F}_\msf{rec},T,\tx_T)$.
    \item Set $success_i\leftarrow0$.
    \item Return a $(\msf{complete},T)$ message to $\mc{Z}$.
   
\end{enumerate}
\end{small}
}
\end{mdframed}
\end{thickframe}
         \caption{The protocol $\single$ for non leader $\party{i}$.}
         \label{fig:single_other}
     \end{figure}
\normalsize
\begin{remark}\label{rem:payment}
As explained in Figure~\ref{fig:single_leader}, the pool leader carries out the payments of the other parties by including them in a special transaction $\tx_T$. Since the exact payment method does not affect our analysis, we do not provide details on the format of $\tx_T$. In practice, each party $\party{i}\neq\party{L}$ could provide $\party{L}$ with a fresh public key $\msf{pk}_i$, and $\party{L}$ would include a payment linked to $\msf{pk}_i$ in $\tx_T$.
\end{remark}
\section{$\single$ as an EVP}\label{sec:equilibrium_proof}

In this section, we provide our main result. Namely, that the protocol $\single$ is an EVP according to Definition~\ref{def:EVP}. In our theorem statement, we quantify over a class of adversaries whose strategy does not result in the mining of blocks that are ``almost'' empty (i.e., they contain only the special payment transaction $\tx_T$). As we shortly explain, restricting to this class is meaningful and does not harm the generality of our result. In particular, we define the following type of adversary.

For some round $T$, let $\{\mbf{V}_1,\ldots,\mbf{V}_{K_T}\}$ be the partition of the party set $\{\party{1},\ldots,\party{n}\}$ such that for every $i\in[K_T]$, the parties in $\mbf{V}_i$ form a pool according to Definition~\ref{def:pool} (trivially, if $\mbf{V}_i$ is a singleton, then the single party in $\mbf{V}_i$ acts on its own). We say that an adversary $\mathcal{A}$ that controls a coalition $\corrupt\subset\{\party{1},\ldots,\party{n}\}$ is \emph{$\mc{O}_\msf{tx}$-respecting} if for every round $T$ and every $\mbf{V}_i$, $i\in[K_T]$, there is at least one party in $\mbf{V}_i$ that asks the transaction oracle $\mc{O}_\msf{tx}$ during $T$. 

We stress that if we lift the above restriction and quantify over all adversaries in our theorem statement, then using similar proof techniques, we can show that a variant of $\single$ where the leader never queries $\mc{O}_\msf{tx}$, ignores its input transactions, and sets the record as the singleton that includes only the special payment transaction (cf. Remark~\ref{rem:variant} for the variant description details), is an EVP. This strategy profile is related to the verifier's dilemma introduced in~\cite{LuuTKS15}, according to which miners are motivated to skip verification of transactions when the cost is significant.  Observe that this variant of $\single$ forms again a single pool with all the parties (which means that is again completely centralized), but violates liveness\footnote{A blockchain protocol satisfies liveness, if every transaction that has been issued and diffused by an honest party  will be included eventually in the ledger with $1-\negl(\secpar)$ probability~\cite{backbone}}. In our main theorem, we do not follow this direction, as for this variant to be functional, it is necessary that there is no external observer that can check the validity of the chain and affect the profit of the parties. This is not true in practice, as external users can easily detect that the blocks are empty, harm the reputation of the system, and thus affect the price of the currency the parties of the pool earn.

Our main theorem statement relies on three reasonable assumptions: (i) the expected rewards per random oracle query are higher than the cost of the query and the cost needed to form the instance for the query (recall that the rewards and the costs are in the same unit), (ii) $\pb=\Omega(\frac{1}{nq})$, and (iii) $\pf<\tfrac{1}{2}$. Moreover, the multiplicative approximation factor is zero. Besides, the three dominant terms in the additive approximation factor are justified as follows:
\begin{enumerate}
    \item[(a)] The term $O\big(\rounds(n-1)\big)\pf\reward$ (a small fraction of the adversary's expected total rewards) appears because the EVP notion compares the exact profit of the adversary in the two executions with $1-\negl(\secpar)$ probability.
    \item[(b)] The term $O\big(\log\secpar\sqrt{\rounds}\big)\cost_\msf{lc}$ (the difference between $\umax(\exec{})$ and $\umin(\honexec)$ in the cost of asking the longest chain oracle) is due to the same reason as above.
   \item[(c)] The term $O\big( \rounds(n-1)\big)\cost_\msf{ltx}$ (the difference between $\umax(\exec{})$ and $\umin(\honexec)$ in the cost of asking the light transaction oracle) derives from the fact that in $\umax(\exec{})$, the parties check that they have got paid correctly by the pool leader. Typically, the cost for this check is relatively small.
\end{enumerate}

\begin{theorem}
\label{th:equilibrium}
Let (i) $\pf\reward>\tfrac{\cost_\msf{lc}+\cost_\msf{fs}+\cost_\msf{tx}}{(1-\frac{\log\secpar}{\sqrt[4]{n}}) \sqrt{n}q}+\cost_\msf{ro}$, (ii) $\pb=\Omega(\frac{1}{nq})$, and (iii) $\pf<\tfrac{1}{2}$. Then, for any $\delta\in \big[\tfrac{\log\secpar}{\sqrt[4]{\rounds n}},1\big)$, the $\single$ protocol is an $(n-1,0,\epsilon')$-EVP according to the utility profit, where
\begin{equation*}
\begin{split}
\epsilon'=&\Big(\big(\tfrac{\log\secpar}{\sqrt{\rounds n}}+\delta\big)\rounds+\log^2\secpar\big(1+\tfrac{1}{\pf}\big) -\big(\tfrac{\log^3\secpar}{\sqrt{\rounds n}}+1+\delta\big)\Big)(n-1)q\pf\reward+\\
&+\tfrac{n-1}{n}\Big((2\log\secpar)\sqrt{\rounds}+1-\tfrac{\log\secpar}{\sqrt{\rounds}}\Big)(1-(1-\pb)^{nq})\cost_\msf{lc}+\\
&+\big(1+\tfrac{\log\secpar}{\sqrt{\rounds}}\big)\rounds(1-(1-\pb)^{nq})(n-1)\cost_\msf{ltx}+\\
&+\tfrac{\log^2\secpar}{n}(\cost_\msf{fs}+\cost_\msf{tx}),
\end{split}    
\end{equation*}
 w.r.t. every   $\mc{O}_\msf{tx}$-respecting adversary $\mc{A}$ and every $\rounds$-admissible environment $\mc{Z}$ that activates the pool leader first in each round.
\end{theorem}
\begin{proof}
We will assume that the adversary has corrupted a set $\corrupt$ with $n-1$ parties and can deviate from the $\single$ protocol arbitrarily. We will prove that for every $\rounds$-admissible environment $\mc{Z}$ that activates the leader first in each round and for every PPT adversary $\mc{A}$ that controls $\corrupt$, it holds that 
\begin{equation*}\label{EVP}
\umax(\exec{}) \leq \umin(\honexec) + \epsilon \cdot \mid \umin(\honexec) \mid +\epsilon'   
\end{equation*}
with overwhelming probability in the security parameter $\secpar$.
\par
Note that we do not quantify over adversaries that control a set $\corrupt'$ with $t'$ parties, where $t'< n-1$. The reason is that for every adversary $\mc{A}'$ that corrupts $\corrupt'$, we can consider an adversary $\mc{A}''$ that corrupts a set $\corrupt''\supset\corrupt'$ with exactly $n-1$ parties and instructs the parties in $\corrupt'$ to deviate exactly like $\mc{A}'$ and the other $n-1-t'$ parties in $\corrupt''\setminus\corrupt'$ to follow the $\single$ protocol. 
\par
At this point, we will describe all the possible deviations of the corrupted parties in the set $\corrupt$ as instructed by $
\mc{A}$. The set $\corrupt$ can include either (i) the pool leader and $n-2$ other members of the pool, or (ii) $n-1$ members of the pool and not the pool leader.

 A round $T$ will be called \textit{payment round} if (i) the array of all blocks that were diffused during the previous round $T-1$ was not empty, and (ii) the array of all the blocks and the set of all the fruits diffused during the previous round $T-1$  do not contain a  block $\hat{\block}:=\langle\langle \hat{h}_{-1},\hat{h}_f,\hat{\nonce},$ $\hat{\msf{dig}},\hat{\record},\hat{h}\rangle,\hat{\mbf{F}}\rangle$
  or a fruit $\hat{\mbf{f}}:=\langle \hat{h}_{-1},\hat{h}_f,\hat{\nonce},\hat{\msf{dig}},\hat{\record},\hat{h}\rangle$, respectively, such that $(\hat{\previous},\hat{h}_f,\hat{\msf{dig}},\hat{\record})\neq(inst_1,inst_2,inst_3,inst_4)$.
%
%

The possible deviations that $\mc{A}$ can perform are any combination of the following strategies.
\begin{enumerate}
    \item[(D1)] $\mc{A}$ instructs a subset of the non leader parties in $\corrupt$ to deviate from step (4) of the $\single$ protocol for one or more rounds, by ignoring the instance received from the pool leader and creating a different instance. Note that this includes the case where the adversary instructs some adversarial parties to abandon the pool. The meaningful possible deviations that we should examine  are the following: 
    \begin{enumerate}
    \item[(i)] the adversarial party ignores the record $\record$ received from the pool leader and updates $inst_4$ with a new record including a transaction that makes the adversarial party as the recipient of the rewards. 
    \item[(ii)] the adversarial party updates $inst_1$ with a hash value of a block different from $h_{-1}$ received from the pool leader. This reflects the scenario where the adversarial party creates a fork.
    \item[(iii)] the adversarial party updates $inst_2$ with the hash of a block that is different from the hash value, $h_f$, of the block received from the pool leader.
    \item[(iv)] the adversarial party updates $inst_3$ with a digest of a fruit set that is different from $\digest(\mbf{F}_\msf{rec})$ received from the pool leader.
    \item[(v)] the adversarial party does not update one or more $inst_i$, $i\in[4]$.
    
    \end{enumerate}
    \item[(D2)] $\mc{A}$ instructs a subset of the non leader parties in $\mathbf{C}$ to ask $\mc{O}_\msf{ltx}$ no queries during one or more payment rounds.
    \item[(D3)] $\mc{A}$ instructs a subset of the parties in $\mathbf{C}$ to ask the oracle $\mc{O}_\msf{ro}$ fewer than $q$ queries during one or more rounds.
    \item[(D4)] $\mc{A}$ instructs a subset of the parties in $\mathbf{C}$ not to send the fruits or the blocks it produces to the Diffuse Functionality for one or more rounds (this is related to block withholding attacks, cf. \cite{ Rosenfeld2011AnalysisOB,7163020}).

    \item[(D5)] $\mc{A}$ instructs a subset of the parties in $\mathbf{C}$ to delay arbitrarily to send the fruits or the blocks it produces to the Diffuse Functionality for one or more rounds.
    \item[(D6)] $\mc{A}$ instructs a subset of the parties in $\mathbf{C}$ to abandon the pool and create a new pool that follows different instructions from the $\single$ protocol.
    \item[(D7)] $\mc{A}$ instructs a subset of the parties in $\mathbf{C}$ to remain in the pool even if it receives a block $\hat{\block}:=\langle\langle \hat{h}_{-1},\hat{h}_f,\hat{\nonce},$ $\hat{\msf{dig}},\hat{\record},\hat{h}\rangle,\hat{\mbf{F}}\rangle$ or a fruit $\langle \hat{h}_{-1},\hat{h}_f,\hat{\nonce},\hat{\msf{dig}},\hat{\record},\hat{h}\rangle$ such that $(\hat{\previous},\hat{h}_f,\hat{\msf{dig}},\hat{\record})\neq(inst_1,inst_2,inst_3,inst_4)$.
    \item[(D8)] $\mc{A}$ instructs a subset of the parties in $\mathbf{C}$ to abandon the pool and follow the $\protocolevp$ protocol.
    \item[(D9)] If $\mathbf{C}$ includes the pool leader, $\mc{A}$ instructs the pool leader to ask no query to $\mc{O}_\msf{fs}$ for one or more rounds.
    \item[(D10)] If $\mathbf{C}$ includes the pool leader, $\mc{A}$ instructs the pool leader to ask no query to $\mc{O}_\msf{tx}$ for one or more rounds.
    \item[(D11)] If $\mathbf{C}$ includes the pool leader, $\mc{A}$ instructs the pool leader to ask no query to $\mc{O}_\msf{lc}$ for one or more payment rounds.
    \item[(D12)] If $\mathbf{C}$ includes the pool leader, $\mc{A}$ instructs the pool leader for one or more payment rounds to  create a special transaction $\tx_T$ that pays the party that does not belong to $\corrupt$ a smaller amount than what described in 3(c) of $\single$ for the pool leader (cf. Figure~\ref{fig:single_leader}). 
   
\end{enumerate}
Note that deviations D1, D2 apply only to the corrupted non leader parties, deviations D3-D8 apply to all parties in $\corrupt$ and deviations D9-D12 apply only to the corrupted pool leader. 

First, we provide a lower bound that $\umin(\honexec)$ achieves with overwhelming probability.

\begin{claim}\label{claim:H}
If (i) $\pf\reward\geq\frac{\cost_\msf{lc}+\cost_\msf{fs}+\cost_\msf{tx}}{(1-\frac{\log\secpar}{\sqrt{n}}) nq}$ and (ii) $\pb=\Omega(\frac{1}{nq})$, then it holds that
\begin{equation*} 
\begin{split}
&\Pr\big[\umin(\honexec)\geq\big(1-\tfrac{\log\secpar}{\sqrt{\rounds n}}\big)(\rounds-\log^2\secpar)(n-1)q\pf\reward-\\
&-\tfrac{n-1}{n}\big(1+\tfrac{\log\secpar}{\sqrt{\rounds}}\big)\rounds(1-(1-\pb)^{nq})(\cost_\msf{lc}+n\cost_\msf{ltx})-\\
&-\big(\tfrac{n-1}{n}\rounds+\tfrac{\log^2\secpar}{n}\big)(\cost_\msf{fs}+\cost_\msf{tx})-\rounds (n-1)q\cost_\msf{ro}\big]\geq\\
&\geq 1-\negl(\secpar).
\end{split}
\end{equation*}
\end{claim}
\textit{Proof of Claim~\ref{claim:H}. }
We say that a block $\block$ is \emph{profitable} if the rewards that derive from the fruits included in $\block$ are higher than the pool leader cost of asking $\cost_\msf{lc},\cost_\msf{fs},\cost_\msf{tx}$ during the mining of $\block$. By the description of $\single$, if $\block$ is profitable, then $\party{L}$ shares the total profit of $\block$ evenly among all $n$ parties of the single pool. Otherwise, $\party{L}$ uses the rewards to cover (part of) the cost of $\block$ while the other parties receive no profit for $\block$.

We show that the probability that a block $\block$ is profitable under the execution $\honexec$ is overwhelming. 
Let $\rho$ be the number of rounds elapsed for mining $\block$. Since the number of the queries the single pool makes per round is $nq$, the number of fruits mined during the mining of $\block$, $Z_0$, follows $\msf{Bin}(\rho nq,\pf)$. By the Chernoff bounds (cf. Appendix~\ref{app:chernoff}),
\begin{equation*}
\Pr\big[Z_0<(1-\tfrac{\log\secpar}{\sqrt{n}})\rho nq\pf\big]\leq e^{-\frac{\log^2\secpar}{2n}\rho nq\pf}=\negl(\secpar).    
\end{equation*}
Thus, with $1-\negl(\secpar)$ probability, the rewards w.r.t. $\block$ are at least $(1-\frac{\log\secpar}{\sqrt{n}})\rho nq\pf\reward$. Besides, the leader cost for $\block$ is $\cost_\msf{lc}+\rho\cost_\msf{fs}+\rho\cost_\msf{tx}$. 
Since $\pf\reward\geq\frac{\cost_\msf{lc}+\cost_\msf{fs}+\cost_\msf{tx}}{(1-\frac{\log\secpar}{\sqrt{n}}) nq}$ and $\rho\geq 1$, we have that with $1-\negl(\secpar)$ probability, it holds that
\[(1-\tfrac{\log\secpar}{\sqrt{n}})\rho nq\pf\reward\geq\rho(\cost_\msf{lc}+\cost_\msf{fs}+\cost_\msf{tx})\geq\cost_\msf{lc}+\rho\cost_\msf{fs}+\rho\cost_\msf{tx},\]
i.e., $\Pr[\block\mbox{ is profitable}]\geq1-\negl(\secpar)$.\\[2pt]
%
%
Next, we define the following random variables. Let $Z$ be the number of mined fruits during $\honexec$ for the first $\rounds-\log^2\secpar$ rounds and $W$ be the number of rounds that at least one block was mined (i.e., the number of calls to $\mc{O}_\msf{lc}$, $\mc{O}_\msf{tlx}$). Since $nq$ mining queries are made per round, $Z\sim\msf{Bin}((\rounds-\log^2\secpar) nq,\pf)$. Besides, the probability that at least one block is produced in some round is $1-(1-\pb)^{nq}$, so $W\sim\msf{Bin}(\rounds,1-(1-\pb)^{nq})$.

Let $t\_reward$ be the total rewards of the single pool and $l\_cost$ be the leader costs in terms of queries to $\mc{O}_\msf{lc},\mc{O}_\msf{fs},\mc{O}_\msf{tx}$ that are shared among the members of the single pool. Let $c\_reward$ be the rewards of the coalition $\corrupt$ and $c\_cost$ be the additional cost that $\corrupt$ incurs besides its share of $l\_cost$.

As shown above, the probability that some block is not profitable is $\negl(\secpar)$. So, by the union bound, the probability that all the blocks are profitable in $\honexec$ is at least $1-\rounds\negl(\secpar)=1-\negl(\secpar)$. As the coalition consists of $n-1$ parties and when all blocks are profitable each party receives an equal profit share, we have that
\begin{equation}\label{eq:c_reward}
\Pr\big[c\_reward=\tfrac{n-1}{n}(t\_reward-l\_cost)\big]=1-\negl(\secpar).
\end{equation}
Moreover, the probability that no block is produced during the last $\log^2\secpar$ rounds is $(1-\pb)^{(\log^2\secpar)nq}=\negl(\secpar)$, for $\pb=\Omega(\frac{1}{nq})$. Thus, all fruits mined during the first $\rounds-\log^2\secpar$ rounds will be included in the chain with $1-\negl(\secpar)$ probability, so
\begin{equation}\label{eq:t_reward}
\Pr[t\_reward\geq Z\reward]=1-\negl(\secpar).    
\end{equation}
For the leader costs that are shared among the pool members, we have that $l\_cost\leq W\cost_\msf{lc}+r_0(\cost_\msf{fs}+\cost_\msf{tx})$, where $r_0$ is the round that the final block was mined in $\honexec$. By the Chernoff bounds, and since $W\sim\msf{Bin}(\rounds,1-(1-\pb)^{nq})$
\begin{equation*}
\begin{split}
&\Pr\big[W\geq\big(1+\tfrac{\log\secpar}{\sqrt{\rounds}}\big)\rounds(1-(1-\pb)^{nq})\big]
=\negl(\secpar).
\end{split}
\end{equation*}
Thus, we have that 
\begin{equation}\label{eq:l_cost}
\begin{split}
\Pr\big[l\_cost&\leq \big(1+\tfrac{\log\secpar}{\sqrt{\rounds}}\big)\rounds(1-(1-\pb)^{nq})\cost_\msf{lc}+r_0(\cost_\msf{fs}+\cost_\msf{tx})\big] =1-\negl(\secpar).
\end{split}
\end{equation}
Each party in the coalition makes also queries to $\cost_\msf{ltx}$ and $\cost_\msf{ro}$. In addition, in case the coalition includes $\party{L}$, we take into account the extra cost $(N-r_0)(\cost_\msf{fs}+\cost_\msf{tx})$ for $\party{L}$ in the last $(N-r_0)$ rounds where no block was produced. In any case, $c\_cost\leq W(n-1)\cost_\msf{ltx}+\rounds(n-1)q\cost_\msf{ro}+(N-r_0)(\cost_\msf{fs}+\cost_\msf{tx})$. By the Chernoff bounds,
\begin{equation}\label{eq:c_cost}
\begin{split}
&\Pr\big[c\_cost\leq \big(1+\tfrac{\log\secpar}{\sqrt{\rounds}}\big)\rounds(1-(1-\pb)^{nq})(n-1)\cost_\msf{ltx}+\\
&+\rounds (n-1)q\cost_\msf{ro}+(N-r_0)(\cost_\msf{fs}+\cost_\msf{tx}) \big]=1-\negl(\secpar).    
\end{split}
\end{equation}
By Eq.~\eqref{eq:c_reward},~\eqref{eq:t_reward},~\eqref{eq:l_cost},~\eqref{eq:c_cost} and for some lower bound $B$ (to be defined), we get that
\begin{equation}\label{eq:H_final}
\begin{split}
&\Pr[\umin(\honexec)\geq B]=\Pr[c\_reward-c\_cost\geq B]\geq \\
\geq&\Pr\big[\tfrac{n-1}{n}(t\_reward-l\_cost)-c\_cost\geq B\big]-\negl(\secpar)\geq\\
\geq&\Pr\big[\tfrac{n-1}{n}\Big(Z\reward-\big(1+\tfrac{\log\secpar}{\sqrt{\rounds}}\big)\rounds(1-(1-\pb)^{nq})\cost_\msf{lc}-r_0(\cost_\msf{fs}+\cost_\msf{tx})\Big)-\\
&-\big(1+\tfrac{\log\secpar}{\sqrt{\rounds}}\big)\rounds(1-(1-\pb)^{nq})(n-1)\cost_\msf{ltx}-\\
&-\rounds (n-1)q\cost_\msf{ro}-(N-r_0)(\cost_\msf{fs}+\cost_\msf{tx})\geq B\big]-\\
&-\negl(\secpar)\geq\\
\geq&\Pr\big[\tfrac{n-1}{n}\Big(Z\reward-\big(1+\tfrac{\log\secpar}{\sqrt{\rounds}}\big)\rounds(1-(1-\pb)^{nq})\cost_\msf{lc}\Big)-\big(\tfrac{n-1}{n}r_0+(\rounds-r_0)\big)(\cost_\msf{fs}+\cost_\msf{tx})-\\
&-\big(1+\tfrac{\log\secpar}{\sqrt{\rounds}}\big)\rounds(1-(1-\pb)^{nq})(n-1)\cost_\msf{ltx}-\rounds (n-1)q\cost_\msf{ro}\geq B\big]-\\
&-\negl(\secpar)=\\
\geq&\Pr\big[Z\reward\geq \tfrac{n}{n-1}B+\big(1+\tfrac{\log\secpar}{\sqrt{\rounds}}\big)\rounds(1-(1-\pb)^{nq})(\cost_\msf{lc}+n\cost_\msf{ltx})+\\
& +\tfrac{n}{n-1}\big(\tfrac{n-1}{n}r_0+(\rounds-r_0)\big)(\cost_\msf{fs}+\cost_\msf{tx})+\rounds nq\cost_\msf{ro}\big]-\negl(\secpar).
\end{split}   
\end{equation}
As already shown, with $1-\negl(\secpar)$ probability, it holds that $r_0\geq\rounds-\log^2\secpar$. Therefore,
\begin{equation}\label{eq:CfsCtx}
\tfrac{n-1}{n}r_0+(\rounds-r_0)=\rounds-\tfrac{1}{n}r_0\leq\tfrac{n-1}{n}\rounds+\tfrac{\log^2\secpar}{n}.
\end{equation}
So, by setting 
\begin{equation*}
\begin{split}
&\tfrac{n}{n-1}B+\big(1+\tfrac{\log\secpar}{\sqrt{\rounds}}\big)\rounds(1-(1-\pb)^{nq})(\cost_\msf{lc}+n\cost_\msf{ltx})+\\
&+\tfrac{n}{n-1}\big(\tfrac{n-1}{n}\rounds+\tfrac{\log^2\secpar}{n}\big)(\cost_\msf{fs}+\cost_\msf{tx})+\rounds nq\cost_\msf{ro}=\\
&=\big(1-\tfrac{\log\secpar}{\sqrt{\rounds n}}\big)(\rounds-\log^2\secpar)nq\pf\reward\Leftrightarrow\\
\Leftrightarrow& B:=\big(1-\tfrac{\log\secpar}{\sqrt{\rounds n}}\big)(\rounds-\log^2\secpar)(n-1)q\pf\reward-\\
&-\tfrac{n-1}{n}\big(1+\tfrac{\log\secpar}{\sqrt{\rounds}}\big)\rounds(1-(1-\pb)^{nq})(\cost_\msf{lc}+n\cost_\msf{ltx})-\\
&-\big(\tfrac{n-1}{n}\rounds+\tfrac{\log^2\secpar}{n}\big)(\cost_\msf{fs}+\cost_\msf{tx})-\rounds (n-1)q\cost_\msf{ro},
\end{split}    
\end{equation*}
and by Eq.~\eqref{eq:H_final},~\eqref{eq:CfsCtx} and the Chernoff bounds, we conclude that
\begin{equation*}
\begin{split}
&\Pr[\umin(\honexec)\geq B]\geq\\
\geq&\Pr\big[Z\geq\big(1-\tfrac{\log\secpar}{\sqrt{\rounds n}}\big)(\rounds-\log^2\secpar)nq\pf\big]-\negl(\secpar)\geq\\
\geq&\big(1-e^{\frac{\log^2\secpar}{\rounds n}(\rounds-\log^2\secpar)nq\pf})-\negl(\secpar)\geq1-\negl(\secpar).
\end{split}    
\end{equation*}
$\quad$\hfill$\dashv$

Next, we provide an upper bound for $\umax(\exec{})$ when $\mc{A}$'s strategy derives by combining deviations D2, D3, D8.

\begin{claim}\label{claim:D2_D3_D8}
Let $\mc{A}$ be an adversary whose strategy comprises a combination of deviations D2, D3, and D8.

If (i) $\pf\reward>\mathrm{max}\Big\{\tfrac{\cost_\msf{lc}+\cost_\msf{fs}+\cost_\msf{tx}}{(1-\frac{\log\secpar}{\sqrt[4]{n}}) \sqrt{n}q},3\big(\tfrac{\cost_\msf{lc}+\cost_\msf{fs}+\cost_\msf{tx}}{(n-1)q}+\cost_\msf{ro}\big)\Big\}$, and (ii) $\pf<\tfrac{1}{2}$, then for any $\delta\in \big[\tfrac{\log\secpar}{\sqrt[4]{\rounds n}},1\big)$, it holds that
\begin{equation*}
\begin{split}
&\Pr\big[\umax(\exec{})\leq (1+\delta)(\rounds-1)(n-1)q\pf\reward+\log^2\secpar(n-1)q\reward-\\
&-\tfrac{n-1}{n}\big(1-\tfrac{\log\secpar}{\sqrt{\rounds}}\big)(\rounds-1)(1-(1-\pb)^{nq})\cost_\msf{lc}-\\
&-\tfrac{n-1}{n}\rounds(\cost_\msf{fs}+\cost_\msf{tx})-\rounds (n-1)q\cost_\msf{ro}\big]\geq1-\negl(\secpar).
\end{split}
\end{equation*}

\end{claim}

\textit{Proof of Claim~\ref{claim:D2_D3_D8}.} Since it comprises a combination of deviations D2, D3, and D8, $\mc{A}$'s strategy can be generally described as follows: Up to some round $r^*$, $\mc{A}$ may instruct the coalition $\corrupt$ to make fewer queries to $\mc{O}_\msf{ltx}$ and $\mc{O}_\msf{ro}$ while remaining members of the single pool. After $r^*$, $\mc{A}$ instructs the coalition to abandon the pool and follow $\protocolevp$ with the difference that the corrupted parties may again make fewer queries to $\mc{O}_\msf{ltx}$ and $\mc{O}_\msf{ro}$.

Let $Q\leq(n-1)q$ be the total number of queries to $\mc{O}_\msf{ro}$ of the coalition per round. Without loss of generality (since we want to upper bound the profit of $\mc{A}$), we assume that the corrupted parties make no queries to $\mc{O}_\msf{ltx}$. We define the following random variables:

 Let $W^-$ be the number of rounds before $r^*$ that at least one block was mined (i.e., the number of calls to $\mc{O}_\msf{lc}$ up to $r^*$). Since the remaining honest party makes $q$ queries to the random oracle per round, the probability that at least one block is produced in some round is $1-(1-\pb)^{q+Q}$, so $W^-\sim\msf{Bin}(r^*-1,1-(1-\pb)^{q+Q})$.  

Let $Z^-$ be the number of fruits mined before $r^*$. Since $q+Q$ queries are made by all parties per round,  $Z^-\sim\msf{Bin}((r^*-1)(q+Q),\pf)$.

Let $t\_reward$ be the total rewards of the single pool up to $r^*$ and $l\_cost$ be the leader costs up to $r^*$ in terms of queries to $\mc{O}_\msf{lc},\mc{O}_\msf{fs},\mc{O}_\msf{tx}$ that are shared among the members of the single pool. Let $c\_reward^-$ be the rewards of $\corrupt$ up to $r^*$ and $c\_cost^-$ be the additional cost up to $r^*$ that $\corrupt$ incurs besides its share of $l\_cost$.

Let $W^+$ be the number of rounds from $r^*$ to $\rounds-1$ that at least one block was mined (i.e., the number of calls to $\mc{O}_\msf{lc}$ after $r^*$). It holds that $W^+\sim\msf{Bin}(\rounds-r^*,1-(1-\pb)^{q+Q})$.

Let $Z^+$ be the number of fruits mined by $\corrupt$ from $r^*$ to $\rounds-1$. Since the parties in $\corrupt$ ask $Q$ queries per round, it holds that $Z^+\sim\msf{Bin}((\rounds-r^*)Q,\pf)$.

Let $c\_reward^+$ be the rewards of $\corrupt$ after $r^*$ and $c\_cost^+$ be the total cost that $\corrupt$ incurs after $r^*$.

Assume that $Q\geq\sqrt{n}q$ (the case where $Q<\sqrt{n}q$ will be studied later). Similarly to Claim~\ref{claim:H}, we can show that if $\pf\reward\geq\frac{\cost_\msf{lc}+\cost_\msf{fs}+\cost_\msf{tx}}{(1-\frac{\log\secpar}{\sqrt[4]{n}}) \sqrt{n}q}$, then with $1-\negl(\secpar)$ probability, all blocks of the execution are profitable. In particular, let $\rho$ be the number of rounds elapsed for mining a block $\block$. Since the number of the queries that all parties make per round is $q+Q$, the number of fruits mined during the mining of $\block$, $Z_0$, follows $\msf{Bin}(\rho (q+Q),\pf)$. By the Chernoff bounds and given that $Q\geq \sqrt{n}q$,
\begin{equation*}
\begin{split}
\Pr\big[Z_0<(1-\tfrac{\log\secpar}{\sqrt[4]{n}})\rho (q+Q)\pf\big]&\leq e^{-\frac{\log^2\secpar}{2\sqrt{n}}\rho (q+Q)\pf}\leq e^{-\frac{\log^2\secpar}{2\sqrt{n}}\rho \sqrt{n}q\pf}=\negl(\secpar). 
\end{split}
\end{equation*}
So, with $1-\negl(\secpar)$ probability, the rewards w.r.t. $\block$ are at least $(1-\frac{\log\secpar}{\sqrt[4]{n}})\rho (q+Q)\pf\reward$. Besides, the leader cost for $\block$ is $\cost_\msf{lc}+\rho\cost_\msf{fs}+\rho\cost_\msf{tx}$. 
Since $\pf\reward\geq\frac{\cost_\msf{lc}+\cost_\msf{fs}+\cost_\msf{tx}}{(1-\frac{\log\secpar}{\sqrt[4]{n}}) \sqrt{n}q}$ and $\rho\geq 1$, we have that with $1-\negl(\secpar)$ probability, it holds that
\begin{equation*}
\begin{split}
&(1-\tfrac{\log\secpar}{\sqrt[4]{n}})\rho (q+Q)\pf\reward>(1-\tfrac{\log\secpar}{\sqrt[4]{n}})\rho \sqrt{n}q\pf\reward\geq\\
\geq&\rho(\cost_\msf{lc}+\cost_\msf{fs}+\cost_\msf{tx})\geq\cost_\msf{lc}+\rho\cost_\msf{fs}+\rho\cost_\msf{tx},
\end{split}
\end{equation*}
i.e., $\Pr[\block\mbox{ is profitable}]\geq1-\negl(\secpar)$.

Thus, since by definition, $t\_reward\leq Z^-\reward$ and $l\_cost=W^-\cost_\msf{lc}+r^*(\cost_\msf{fs}+\cost_\msf{tx})$, and given that $\corrupt$ has $n-1$ parties, we have that
with $1-\negl(\secpar)$ probability,
\begin{equation}\label{eq:c_reward-}
c\_reward^-\leq\tfrac{n-1}{n}\big(Z^-\reward-W^-\cost_\msf{lc}+r^*(\cost_\msf{fs}+\cost_\msf{tx})\big).
\end{equation}
Besides, we directly get that 
\begin{equation}\label{eq:c_cost-}
c\_cost^-\geq 0\cost_\msf{ltx}+r^*Q\cost_\msf{ro}=r^*Q\cost_\msf{ro}.
\end{equation}
Upon abandoning the pool, for the coalition $\corrupt$ it holds that
\begin{equation}\label{eq:c_reward+}
c\_reward^+\leq Z^+\reward
\end{equation}
\begin{equation}\label{eq:c_cost+}
c\_cost^+\geq W^+\cost_\msf{lc}+(\rounds-r^*)(\cost_\msf{fs}+\cost_\msf{tx})-(\rounds-r^*)Q\cost_\msf{ro}
\end{equation}
The above lower bound for $c\_cost^+$ holds because $\mc{A}$ follows a combination of D2, D3, and D8, so for every round after $r^*$, there is at least one corrupted party that interacts with $\mc{O}_\msf{lc},\mc{O}_\msf{fs},\mc{O}_\msf{tx}$ according to $\protocolevp$ (on behalf of $\corrupt$).

By Eq.~\eqref{eq:c_reward-},~\eqref{eq:c_cost-},~\eqref{eq:c_reward+},~\eqref{eq:c_cost+} and for lower bound $B$ (to be defined), we have that
\begin{equation*}
\begin{split}
&\Pr[\umax(\exec{})\geq B]=\\
=&\Pr[(c\_reward^- - c\_cost^-)+(c\_reward^+-c\_cost^+)\geq B]\leq\\
\leq&\Pr\big[(\tfrac{n-1}{n}Z^-+Z^+)\reward-(\tfrac{n-1}{n}W^-+W^+)\cost_\msf{lc}-\\
&-\big(\tfrac{n-1}{n}r^*(\cost_\msf{fs}+\cost_\msf{tx})+(\rounds-r^*)(\cost_\msf{fs}+\cost_\msf{tx})\big)-\rounds Q\cost_\msf{ro}\geq B\big]+\negl(\secpar).\\
\leq&\Pr\big[(\tfrac{n-1}{n}Z^-+Z^+)\reward-\tfrac{n-1}{n}(W^-+W^+)\cost_\msf{lc}-\tfrac{n-1}{n}\rounds(\cost_\msf{fs}+\cost_\msf{tx})-\rounds Q\cost_\msf{ro}\geq B\big]+\negl(\secpar).
\end{split}
\end{equation*}
Now observe that the random variable $W^-+W^+$ follows $\msf{Bin}((r^*-1)+(\rounds-r^*),1-(1-\pb)^{q+Q})$, i.e., $W^-+W^+\sim\msf{Bin}(\rounds-1,1-(1-\pb)^{q+Q})$. So, by the Chernoff bounds,
\begin{equation*}
\Pr\big[W^-+W^+\leq\big(1-\tfrac{\log\secpar}{\sqrt{\rounds}}\big)(\rounds-1)(1-(1-\pb)^{q+Q})\big]=\negl(\secpar).
\end{equation*}
Hence, we have that
\begin{equation}\label{eq:bound_all_cases}
\begin{split}
&\Pr[\umax(\exec{})\geq B]\leq\Pr\big[(\tfrac{n-1}{n}Z^-+Z^+)\reward\geq B+\\
&+\tfrac{n-1}{n}\big(1-\tfrac{\log\secpar}{\sqrt{\rounds}}\big)(\rounds-1)(1-(1-\pb)^{q+Q})\cost_\msf{lc}+\tfrac{n-1}{n}\rounds(\cost_\msf{fs}+\cost_\msf{tx})+\rounds Q\cost_\msf{ro}\big]+\negl(\secpar).
\end{split}
\end{equation}
We study the following cases for the value $r^*$:

\textbf{Case 1:} $r^*<\log^2\secpar$. Since $Z^-\leq(r^*-1)(q+Q)$ and by Eq.~\eqref{eq:bound_all_cases},
\begin{equation*}
\begin{split}
&\Pr[\umax(\exec{})\geq B]\leq\Pr\big[Z^+\reward\geq B+\\
&+\tfrac{n-1}{n}\big(1-\tfrac{\log\secpar}{\sqrt{\rounds}}\big)(\rounds-1)(1-(1-\pb)^{q+Q})\cost_\msf{lc}+\\
&+\tfrac{n-1}{n}\rounds(\cost_\msf{fs}+\cost_\msf{tx})+\rounds Q\cost_\msf{ro}-\tfrac{n-1}{n}(r^*-1)(q+Q)\reward\big]+\\
&+\negl(\secpar).
\end{split}
\end{equation*}
To apply the Chernoff bounds for $Z^+$, we want to set $B$ such that for every $Q$, it holds that
\begin{equation*}
\begin{split}
& B+\tfrac{n-1}{n}\big(1-\tfrac{\log\secpar}{\sqrt{\rounds}}\big)(\rounds-1)(1-(1-\pb)^{q+Q})\cost_\msf{lc}+\\
&+\tfrac{n-1}{n}\rounds(\cost_\msf{fs}+\cost_\msf{tx})+\rounds Q\cost_\msf{ro}-\\
&-\tfrac{n-1}{n}(r^*-1)(q+Q)\reward\geq\big(1+\tfrac{\log\secpar}{\sqrt[4]{\rounds n}}\big)(\rounds-r^*)Q\pf\reward\Leftrightarrow\\
\Leftrightarrow&B\geq\Big(\tfrac{n-1}{n}\big(1-\tfrac{\log\secpar}{\sqrt{\rounds}}\big)(\rounds-1)(1-\pb)^q\cost_\msf{lc}\Big)(1-\pb)^Q+\\
+&\Big(\big(1+\tfrac{\log\secpar}{\sqrt[4]{\rounds n}}\big)(\rounds-r^*)\pf\reward+\tfrac{n-1}{n}(r^*-1)\reward-\rounds \cost_\msf{ro}\Big)Q+\\
+&\tfrac{n-1}{n}(r^*-1)q\reward-\tfrac{n-1}{n}\big(1-\tfrac{\log\secpar}{\sqrt{\rounds}}\big)(\rounds-1)\cost_\msf{lc}-\tfrac{n-1}{n}\rounds(\cost_\msf{fs}+\cost_\msf{tx}).
\end{split}
\end{equation*}
We observe that the right term of the above inequality can be expressed as function of $Q$ of the form $f(Q)=a\cdot x^Q+b\cdot Q+c$. 

Next, we show that, if $\pf\reward>2\cost_\msf{ro}$,
then $f(Q)$ has a maximum at $(n-1)q$ in the range  $[0,(n-1)q]$, i.e,. when the coalition asks all available queries. In particular, we want to set $B$ as an upper bound of
\begin{equation*}
\begin{split}
&\Big(\tfrac{n-1}{n}\big(1-\tfrac{\log\secpar}{\sqrt{\rounds}}\big)(\rounds-1)(1-\pb)^q\cost_\msf{lc}\Big)(1-\pb)^Q+\\
+&\Big(\big(1+\tfrac{\log\secpar}{\sqrt[4]{\rounds n}}\big)(\rounds-r^*)\pf\reward+\tfrac{n-1}{n}(r^*-1)\reward-\rounds \cost_\msf{ro}\Big)Q+\\
+&\tfrac{n-1}{n}(r^*-1)q\reward-\tfrac{n-1}{n}\big(1-\tfrac{\log\secpar}{\sqrt{\rounds}}\big)(\rounds-1)\cost_\msf{lc}-\tfrac{n-1}{n}\rounds(\cost_\msf{fs}+\cost_\msf{tx}).
\end{split}
\end{equation*}
To do so, we study the function $f(Q)=a\cdot x^Q+b\cdot Q+c$, where
\begin{align*}
x&=1-\pb\\
    a&=\tfrac{n-1}{n}\big(1-\tfrac{\log\secpar}{\sqrt{\rounds}}\big)(\rounds-1)(1-\pb)^q\cost_\msf{lc}\\
    b&=\big(1+\tfrac{\log\secpar}{\sqrt[4]{\rounds n}}\big)(\rounds-r^*)\pf\reward+\tfrac{n-1}{n}(r^*-1)\reward-\rounds \cost_\msf{ro}\\
    c&=\tfrac{n-1}{n}(r^*-1)q\reward-\tfrac{n-1}{n}\big(1-\tfrac{\log\secpar}{\sqrt{\rounds}}\big)(\rounds-1)\cost_\msf{lc}-\tfrac{n-1}{n}\rounds(\cost_\msf{fs}+\cost_\msf{tx})
\end{align*}
If $\pf\reward>2\cost_\msf{ro}>\tfrac{Nn}{(N-1)(n-1)}\cost_\msf{ro}$, then it is easy to see that 
\begin{equation*}
\begin{split}
b&=\big(1+\tfrac{\log\secpar}{\sqrt[4]{\rounds n}}\big)(\rounds-r^*)\pf\reward+\tfrac{n-1}{n}(r^*-1)\reward-\rounds \cost_\msf{ro}>\\
&>\tfrac{n-1}{n}(\rounds-r^*)\pf\reward+\tfrac{n-1}{n}(r^*-1)\pf\reward-\rounds \cost_\msf{ro}=\\
&=\tfrac{n-1}{n}(\rounds-1)\pf\reward-\rounds \cost_\msf{ro}>0.
\end{split}    
\end{equation*}
In order to find the maximum of $f(Q)$ for $Q\in[0,(n-1)q]$, we compute
\begin{equation*}
\begin{split}
&f'(Q)=0\Rightarrow  a\cdot\ln x\cdot x^Q+b=0\Rightarrow Q=\dfrac{\ln\big(\frac{b}{a\cdot\ln (1/x)}\big)}{\ln x}
\end{split}    
\end{equation*}
Since $b>0$ and $\ln x<0$, we have that $f'$ is increasing. Thus, $\frac{\ln\big(\frac{b}{a\cdot\ln (1/x)}\big)}{\ln x}$ is a minimum for $f$. In addition, $p_b$ is typically a small value so $x$ is close to $1$. Consequently, we may assume that $\ln (1/x)$ is sufficiently small so that $\frac{b}{a\cdot\ln (1/x)}>1$. The latter implies that $\frac{\ln\big(\frac{b}{a\cdot\ln (1/x)}\big)}{\ln x}<0$, so given that $f'$ is increasing, we get that $f'(Q)>0$ for $Q\in[0,(n-1)q]$. Therefore, the maximum of $f$ in $[0,(n-1)q]$ is $(n-1)q$. 

By the above, we have that
\begin{equation*}
\begin{split}
f(Q)&\leq\Big(\tfrac{n-1}{n}\big(1-\tfrac{\log\secpar}{\sqrt{\rounds}}\big)(\rounds-1)(1-\pb)^q\cost_\msf{lc}\Big)(1-\pb)^{(n-1)q}+\\
&\quad+\Big(\big(1+\tfrac{\log\secpar}{\sqrt[4]{\rounds n}}\big)(\rounds-r^*)\pf\reward+\tfrac{n-1}{n}(r^*-1)\reward-\rounds \cost_\msf{ro}\Big)(n-1)q+\\
&\quad+\tfrac{n-1}{n}(r^*-1)q\reward-\tfrac{n-1}{n}\big(1-\tfrac{\log\secpar}{\sqrt{\rounds}}\big)(\rounds-1)\cost_\msf{lc}-\tfrac{n-1}{n}\rounds(\cost_\msf{fs}+\cost_\msf{tx})=\\
&=\big(1+\tfrac{\log\secpar}{\sqrt[4]{\rounds n}}\big)(\rounds-r^*)(n-1)q\pf\reward+(r^*-1)(n-1)q\reward-\\
&\quad-\tfrac{n-1}{n}\big(1-\tfrac{\log\secpar}{\sqrt{\rounds}}\big)(\rounds-1)(1-(1-\pb)^{nq})\cost_\msf{lc}-\tfrac{n-1}{n}\rounds(\cost_\msf{fs}+\cost_\msf{tx})-\rounds (n-1)q\cost_\msf{ro}.
\end{split}
\end{equation*}
Moreover, given that $\pf<\frac{1}{2}<\frac{1}{1+\tfrac{\log\secpar}{\sqrt[4]{\rounds n}}}$ and $r^*<\log^2\secpar$, we have that 
\begin{equation*}
\begin{split}
&\big(1+\tfrac{\log\secpar}{\sqrt[4]{\rounds n}}\big)(\rounds-r^*)(n-1)q\pf\reward+(r^*-1)(n-1)q\reward=\\
=&\big(1+\tfrac{\log\secpar}{\sqrt[4]{\rounds n}}\big)\rounds(n-1)q\pf\reward+\big(1-\big(1+\tfrac{\log\secpar}{\sqrt[4]{\rounds n}}\big)\pf\big)(n-1)q\reward r^*-(n-1)q\reward<\\
<&\big(1+\tfrac{\log\secpar}{\sqrt[4]{\rounds n}}\big)\rounds(n-1)q\pf\reward+\big(1-\big(1+\tfrac{\log\secpar}{\sqrt[4]{\rounds n}}\big)\pf\big)(n-1)q\reward \log^2\secpar-(n-1)q\reward=\\
=&\big(1+\tfrac{\log\secpar}{\sqrt[4]{\rounds n}}\big)(\rounds-\log^2\secpar)(n-1)q\pf\reward+(\log^2\secpar-1)(n-1)q\reward.
\end{split}
\end{equation*}
Therefore, we set the upper bound for $f(Q)$ as
\begin{equation}\label{eq:bound_case_1}
\begin{split}
B&=\big(1+\tfrac{\log\secpar}{\sqrt[4]{\rounds n}}\big)(\rounds-\log^2\secpar)(n-1)q\pf\reward+(\log^2\secpar-1)(n-1)q\reward-\\
&\quad-\tfrac{n-1}{n}\big(1-\tfrac{\log\secpar}{\sqrt{\rounds}}\big)(\rounds-1)(1-(1-\pb)^{nq})\cost_\msf{lc}-\tfrac{n-1}{n}\rounds(\cost_\msf{fs}+\cost_\msf{tx})-\rounds (n-1)q\cost_\msf{ro}.
\end{split}
\end{equation}
For this value of $B$ and by the Chernoff bounds, we have that %
\begin{equation*}
\begin{split}
&\Pr[\umax(\exec{})\geq B]\leq\\
\leq&\Pr\big[Z^+\geq\big(1+\tfrac{\log\secpar}{\sqrt[4]{\rounds n}}\big)(\rounds-r^*)Q\pf\big]+\negl(\secpar)\leq\\
\leq&e^{-\frac{\log^2\secpar}{3\sqrt{\rounds n}}(\rounds-r^*)Q\pf}+\negl(\secpar)\leq\\
\leq&e^{-\frac{\log^2\secpar}{3\sqrt{\rounds n}}(\rounds-\log^2\secpar)\sqrt{n}q\pf}+\negl(\secpar)\leq\negl(\secpar).
\end{split}
\end{equation*}

\textbf{Case 2:} $\rounds-r^*<\log^2\secpar$. By the fact that $Z^+<(\rounds-r^*)Q$, by Eq.~\eqref{eq:bound_all_cases}, we get that  
\begin{equation*}
\begin{split}
&\Pr[\umax(\exec{})\geq B]\leq\Pr\big[\tfrac{n-1}{n}Z^-\reward\geq B+\\
&+\tfrac{n-1}{n}\big(1-\tfrac{\log\secpar}{\sqrt{\rounds}}\big)(\rounds-1)(1-(1-\pb)^{q+Q})\cost_\msf{lc}+\\
&+\tfrac{n-1}{n}\rounds(\cost_\msf{fs}+\cost_\msf{tx})+\rounds Q\cost_\msf{ro}-(\rounds-r^*)Q\reward\big]+\\
&+\negl(\secpar).
\end{split}
\end{equation*}
To apply the Chernoff bounds for $Z^-$, we want to set $B$ such that for every $Q$, it holds that
\begin{equation*}
\begin{split}
& \tfrac{n}{n-1}B+\big(1-\tfrac{\log\secpar}{\sqrt{\rounds}}\big)(\rounds-1)(1-(1-\pb)^{q+Q})\cost_\msf{lc}+\\
&+\rounds(\cost_\msf{fs}+\cost_\msf{tx})+\tfrac{n}{n-1}\rounds Q\cost_\msf{ro}-\\
&-\tfrac{n}{n-1}(\rounds-r^*)Q\reward\geq\big(1+\tfrac{\log\secpar}{\sqrt[4]{\rounds n}}\big)(r^*-1)(q+Q)\pf\reward\Leftrightarrow\\
\Leftrightarrow&B\geq\Big(\tfrac{n-1}{n}\big(1-\tfrac{\log\secpar}{\sqrt{\rounds}}\big)(\rounds-1)(1-\pb)^q\cost_\msf{lc}\Big)(1-\pb)^Q+\\
+&\Big(\tfrac{n-1}{n}\big(1+\tfrac{\log\secpar}{\sqrt[4]{\rounds n}}\big)(r^*-1)\pf\reward+(\rounds-r^*)\reward-\rounds \cost_\msf{ro}\Big)Q+\\
+&\tfrac{n-1}{n}\big(1+\tfrac{\log\secpar}{\sqrt[4]{\rounds n}}\big)(r^*-1)q\pf\reward-\tfrac{n-1}{n}\big(1-\tfrac{\log\secpar}{\sqrt{\rounds}}\big)(\rounds-1)\cost_\msf{lc}-\tfrac{n-1}{n}\rounds(\cost_\msf{fs}+\cost_\msf{tx}).
\end{split}
\end{equation*}
Namely, we want to set $B$ as an upper bound of
\begin{equation*}
\begin{split}
&\Big(\tfrac{n-1}{n}\big(1-\tfrac{\log\secpar}{\sqrt{\rounds}}\big)(\rounds-1)(1-\pb)^q\cost_\msf{lc}\Big)(1-\pb)^Q+\\
+&\Big(\tfrac{n-1}{n}\big(1+\tfrac{\log\secpar}{\sqrt[4]{\rounds n}}\big)(r^*-1)\pf\reward+(\rounds-r^*)\reward-\rounds \cost_\msf{ro}\Big)Q+\\
+&\tfrac{n-1}{n}\big(1+\tfrac{\log\secpar}{\sqrt[4]{\rounds n}}\big)(r^*-1)q\pf\reward-\tfrac{n-1}{n}\big(1-\tfrac{\log\secpar}{\sqrt{\rounds}}\big)(\rounds-1)\cost_\msf{lc}-\tfrac{n-1}{n}\rounds(\cost_\msf{fs}+\cost_\msf{tx}).
\end{split}
\end{equation*}
To do so, we study the function $g(Q)=a'\cdot x^Q+b'\cdot Q+c'$, where
\begin{align*}
x&=1-\pb\\
    a'&=\tfrac{n-1}{n}\big(1-\tfrac{\log\secpar}{\sqrt{\rounds}}\big)(\rounds-1)(1-\pb)^q\cost_\msf{lc}\\
    b'&=\tfrac{n-1}{n}\big(1+\tfrac{\log\secpar}{\sqrt[4]{\rounds n}}\big)(r^*-1)\pf\reward+(\rounds-r^*)\reward-\rounds \cost_\msf{ro}\\
    c'&=\tfrac{n-1}{n}\big(1+\tfrac{\log\secpar}{\sqrt[4]{\rounds n}}\big)(r^*-1)q\pf\reward-\tfrac{n-1}{n}\big(1-\tfrac{\log\secpar}{\sqrt{\rounds}}\big)(\rounds-1)\cost_\msf{lc}-\tfrac{n-1}{n}\rounds(\cost_\msf{fs}+\cost_\msf{tx})
\end{align*}
If $\pf\reward>2\cost_\msf{ro}>\tfrac{Nn}{(N-1)(n-1)}\cost_\msf{ro}$, then it is easy to see that 
\begin{equation*}
\begin{split}
b'&=\tfrac{n-1}{n}\big(1+\tfrac{\log\secpar}{\sqrt[4]{\rounds n}}\big)(r^*-1)\pf\reward+(\rounds-r^*)\reward-\rounds \cost_\msf{ro}>\\
&>\tfrac{n-1}{n})(r^*-1)\pf\reward+(\rounds-r^*)\pf\reward-\rounds \cost_\msf{ro}>\\
&>\tfrac{n-1}{n})(\rounds-1)\pf\reward-\rounds \cost_\msf{ro}>0.
\end{split}    
\end{equation*}
In order to find the maximum of $g(Q)$ for $Q\in[0,(n-1)q]$, we compute
\begin{equation*}
\begin{split}
&g'(Q)=0\Rightarrow  a'\cdot\ln x\cdot x^Q+b'=0\Rightarrow Q=\dfrac{\ln\big(\frac{b'}{a'\cdot\ln (1/x)}\big)}{\ln x}
\end{split}    
\end{equation*}
Just like function $f$ in Case 1, we can conclude that the maximum of $g$ in $[0,(n-1)q]$ is $(n-1)q$. Thus, we have that
\begin{equation*}
\begin{split}
g(Q)&\leq\Big(\tfrac{n-1}{n}\big(1-\tfrac{\log\secpar}{\sqrt{\rounds}}\big)(\rounds-1)(1-\pb)^q\cost_\msf{lc}\Big)(1-\pb)^{(n-1)q}+\\
&\quad+\Big(\tfrac{n-1}{n}\big(1+\tfrac{\log\secpar}{\sqrt[4]{\rounds n}}\big)(r^*-1)\pf\reward+(\rounds-r^*)\reward-\rounds \cost_\msf{ro}\Big)(n-1)q+\\
&\quad+\tfrac{n-1}{n}\big(1+\tfrac{\log\secpar}{\sqrt[4]{\rounds n}}\big)(r^*-1)q\pf\reward-\tfrac{n-1}{n}\big(1-\tfrac{\log\secpar}{\sqrt{\rounds}}\big)(\rounds-1)\cost_\msf{lc}-\tfrac{n-1}{n}\rounds(\cost_\msf{fs}+\cost_\msf{tx})=\\
&=\big(1+\tfrac{\log\secpar}{\sqrt[4]{\rounds n}}\big)(r^*-1)(n-1)q\pf\reward+(\rounds-r^*)(n-1)q\reward-\\
&\quad-\tfrac{n-1}{n}\big(1-\tfrac{\log\secpar}{\sqrt{\rounds}}\big)(\rounds-1)(1-(1-\pb)^{nq})\cost_\msf{lc}-\tfrac{n-1}{n}\rounds(\cost_\msf{fs}+\cost_\msf{tx})-\rounds (n-1)q\cost_\msf{ro}.
\end{split}
\end{equation*}
Given that $\pf<\frac{1}{2}<\frac{1}{1+\tfrac{\log\secpar}{\sqrt[4]{\rounds n}}}$ and $\rounds-r^*<\log^2\secpar$, we have that 
\begin{equation*}
\begin{split}
&\big(1+\tfrac{\log\secpar}{\sqrt[4]{\rounds n}}\big)(r^*-1)(n-1)q\pf\reward+(\rounds-r^*)(n-1)q\reward=\\
=&\rounds(n-1)q\reward-\big(1-\big(1+\tfrac{\log\secpar}{\sqrt[4]{\rounds n}}\big)\pf\big)(n-1)q\reward r^*-\big(1+\tfrac{\log\secpar}{\sqrt[4]{\rounds n}}\big)(n-1)q\pf\reward<\\
<&\rounds(n-1)q\reward-\big(1-\big(1+\tfrac{\log\secpar}{\sqrt[4]{\rounds n}}\big)\pf\big)(n-1)q\reward (\rounds-\log^2\secpar)-\big(1+\tfrac{\log\secpar}{\sqrt[4]{\rounds n}}\big)(n-1)q\pf\reward=\\
=&\big(1+\tfrac{\log\secpar}{\sqrt[4]{\rounds n}}\big)(\rounds-\log^2\secpar-1)(n-1)q\pf\reward+\log^2\secpar(n-1)q\reward.
\end{split}
\end{equation*}
Therefore, we set the upper bound for $g(Q)$ as
\begin{equation}\label{eq:bound_case_2}
\begin{split}
B&=\big(1+\tfrac{\log\secpar}{\sqrt[4]{\rounds n}}\big)(\rounds-\log^2\secpar-1)(n-1)q\pf\reward+\log^2\secpar(n-1)q\reward-\\
&\quad-\tfrac{n-1}{n}\big(1-\tfrac{\log\secpar}{\sqrt{\rounds}}\big)(\rounds-1)(1-(1-\pb)^{nq})\cost_\msf{lc}-\tfrac{n-1}{n}\rounds(\cost_\msf{fs}+\cost_\msf{tx})-\rounds (n-1)q\cost_\msf{ro}.
\end{split}
\end{equation}
For this value of $B$ and by the Chernoff bounds, we have that %
\begin{equation*}
\begin{split}
&\Pr[\umax(\exec{})\geq B]\leq\\
\leq&\Pr\big[Z^-\geq\big(1+\tfrac{\log\secpar}{\sqrt[4]{\rounds n}}\big)(r^*-1)(q+Q)\pf\big]+\negl(\secpar)\leq\negl(\secpar).
\end{split}
\end{equation*}

%

\textbf{Case 3:} $\log^2\secpar\leq r^*\leq \rounds-\log^2\secpar$. In this case, by the Chernoff bounds, we have that for $\delta\in(0,1)$
\begin{align*}
\Pr\big[Z^-\geq\big(1+\delta\big)(r^*-1)(q+Q)\pf\big]
=\negl(\secpar).\\
\Pr\big[Z^+\geq\big(1+\delta\big)(\rounds-r^*)Q\pf\big]
=\negl(\secpar).
\end{align*}
By the above, with $1-\negl(\secpar)$ probability, it holds that 
\begin{equation*}
\begin{split}&\tfrac{n-1}{n}Z^- + Z^+<\\
<&\Big(\tfrac{n-1}{n}(r^*-1)(q+Q)+(\rounds-r^*)Q\Big)(1+\delta)\pf=\\
=&\Big(\big(\tfrac{n-1}{n}(q+Q)-Q\big)r^*-\tfrac{n-1}{n}(q+Q)+\rounds Q\Big)(1+\delta)\pf.
\end{split}
\end{equation*}
Since $Q\leq(n-1)q$, it holds that $\tfrac{n-1}{n}(q+Q)-Q\geq0$. So, given that $r^*\leq\rounds$, we have that with $1-\negl(\secpar)$ probability
\begin{equation*}
\begin{split}&\tfrac{n-1}{n}Z^- + Z^+<
\tfrac{n-1}{n}(\rounds-1)(q+Q)(1+\delta)\pf.
\end{split}
\end{equation*}
Thus, by Eq.~\eqref{eq:bound_all_cases}, we set $B$ such that for every $Q$, it holds 
\begin{equation*}
\begin{split}
& B+\tfrac{n-1}{n}\big(1-\tfrac{\log\secpar}{\sqrt{\rounds}}\big)(\rounds-1)(1-(1-\pb)^{q+Q})\cost_\msf{lc}+\\
&+\tfrac{n-1}{n}\rounds(\cost_\msf{fs}+\cost_\msf{tx})+\rounds Q\cost_\msf{ro}\geq\\
&\geq(1+\delta)\tfrac{n-1}{n}(\rounds-1)(q+Q)\pf\reward\Leftrightarrow\\
\Leftrightarrow&B\geq\Big(\tfrac{n-1}{n}\big(1-\tfrac{\log\secpar}{\sqrt{\rounds}}\big)(\rounds-1)(1-\pb)^q\cost_\msf{lc}\Big)(1-\pb)^Q+\\
+&\Big((1+\delta)\tfrac{n-1}{n}(\rounds-1)\pf\reward-\rounds \cost_\msf{ro}\Big)Q+\\
+&(1+\delta)\tfrac{n-1}{n}(\rounds-1)q\pf\reward -\tfrac{n-1}{n}\big(1-\tfrac{\log\secpar}{\sqrt{\rounds}}\big)(\rounds-1)\cost_\msf{lc}-\tfrac{n-1}{n}\rounds(\cost_\msf{fs}+\cost_\msf{tx}).
\end{split}
\end{equation*}
Namely, we want to set $B$ as an upper bound of
\begin{equation*}
\begin{split}
&\Big(\tfrac{n-1}{n}\big(1-\tfrac{\log\secpar}{\sqrt{\rounds}}\big)(\rounds-1)(1-\pb)^q\cost_\msf{lc}\Big)(1-\pb)^Q+\\
+&\Big((1+\delta)\tfrac{n-1}{n}(\rounds-1)\pf\reward-\rounds \cost_\msf{ro}\Big)Q+\\
+&(1+\delta)\tfrac{n-1}{n}(\rounds-1)q\pf\reward -\tfrac{n-1}{n}\big(1-\tfrac{\log\secpar}{\sqrt{\rounds}}\big)(\rounds-1)\cost_\msf{lc}-\tfrac{n-1}{n}\rounds(\cost_\msf{fs}+\cost_\msf{tx}).
\end{split}
\end{equation*}
To do so, we study the function $h(Q)=a''\cdot x^Q+b''\cdot Q+c''$, where
\begin{align*}
x&=1-\pb\\
    a''&=\tfrac{n-1}{n}\big(1-\tfrac{\log\secpar}{\sqrt{\rounds}}\big)(\rounds-1)(1-\pb)^q\cost_\msf{lc}\\
    b''&=(1+\delta)\tfrac{n-1}{n}(\rounds-1)\pf\reward-\rounds \cost_\msf{ro}\\
    c''&=(1+\delta)\tfrac{n-1}{n}(\rounds-1)q\pf\reward-\tfrac{n-1}{n}\big(1-\tfrac{\log\secpar}{\sqrt{\rounds}}\big)(\rounds-1)\cost_\msf{lc}-\tfrac{n-1}{n}\rounds(\cost_\msf{fs}+\cost_\msf{tx})
\end{align*}
If $\pf\reward>2\cost_\msf{ro}>\tfrac{Nn}{(N-1)(n-1)}\cost_\msf{ro}$, then it is easy to see that 
\begin{equation*}
\begin{split}
b''&=(1+\delta)\tfrac{n-1}{n}(\rounds-1)\pf\reward-\rounds \cost_\msf{ro}>0.
\end{split}    
\end{equation*}
In order to find the maximum of $h(Q)$ for $Q\in[0,(n-1)q]$, we compute
\begin{equation*}
\begin{split}
&h'(Q)=0\Rightarrow  a''\cdot\ln x\cdot x^Q+b''=0\Rightarrow Q=\dfrac{\ln\big(\frac{b''}{a''\cdot\ln (1/x)}\big)}{\ln x}
\end{split}    
\end{equation*}
Just like function $f$ in Case 1, we can conclude that the maximum of $h$ in $[0,(n-1)q]$ is $(n-1)q$. Thus, we have that
\begin{equation*}
\begin{split}
h(Q)&\leq\Big(\tfrac{n-1}{n}\big(1-\tfrac{\log\secpar}{\sqrt{\rounds}}\big)(\rounds-1)(1-\pb)^q\cost_\msf{lc}\Big)(1-\pb)^{(n-1)q}+\\
&\quad+\Big((1+\delta)\tfrac{n-1}{n}(\rounds-1)\pf\reward-\rounds \cost_\msf{ro}\Big)(n-1)q+\\
&\quad+(1+\delta)\tfrac{n-1}{n}(\rounds-1)q\pf\reward -\tfrac{n-1}{n}\big(1-\tfrac{\log\secpar}{\sqrt{\rounds}}\big)(\rounds-1)\cost_\msf{lc}-\tfrac{n-1}{n}\rounds(\cost_\msf{fs}+\cost_\msf{tx})=\\
&=(1+\delta)(\rounds-1)(n-1)q\pf\reward-\\
&\quad-\tfrac{n-1}{n}\big(1-\tfrac{\log\secpar}{\sqrt{\rounds}}\big)(\rounds-1)(1-(1-\pb)^{nq})\cost_\msf{lc}-\tfrac{n-1}{n}\rounds(\cost_\msf{fs}+\cost_\msf{tx})-\rounds (n-1)q\cost_\msf{ro}.
\end{split}
\end{equation*}
Therefore, we set the upper bound for $h(Q)$ as
\begin{equation}\label{eq:bound_case_3}
\begin{split}
B&=(1+\delta)(\rounds-1)(n-1)q\pf\reward-\\
&\quad-\tfrac{n-1}{n}\big(1-\tfrac{\log\secpar}{\sqrt{\rounds}}\big)(\rounds-1)(1-(1-\pb)^{nq})\cost_\msf{lc}-\tfrac{n-1}{n}\rounds(\cost_\msf{fs}+\cost_\msf{tx})-\rounds (n-1)q\cost_\msf{ro}.
\end{split}
\end{equation}
For this value of $B$, we get $\Pr[\umax(\exec{})\geq B]=\negl(\secpar).$
\smallskip

Given Cases 1,2, and 3, we provide a final bound that dominates all three upper bounds in Eq.~\eqref{eq:bound_case_1},~\eqref{eq:bound_case_2}, and~\eqref{eq:bound_case_3}, respectively. In particular, for any $\delta\in \big[\tfrac{\log\secpar}{\sqrt[4]{\rounds n}},1\big)$, we set
\begin{equation}\label{eq:final_bound}
\begin{split}
B&=(1+\delta)(\rounds-1)(n-1)q\pf\reward+\log^2\secpar(n-1)q\reward-\\
&\quad-\tfrac{n-1}{n}\big(1-\tfrac{\log\secpar}{\sqrt{\rounds}}\big)(\rounds-1)(1-(1-\pb)^{nq})\cost_\msf{lc}-\tfrac{n-1}{n}\rounds(\cost_\msf{fs}+\cost_\msf{tx})-\rounds (n-1)q\cost_\msf{ro}.
\end{split}
\end{equation}
Clearly, the above bound dominates the ones in Eq.~\eqref{eq:bound_case_1},~\eqref{eq:bound_case_2}, and and~\eqref{eq:bound_case_3}. Thus, for this value of $B$, we conclude that %
\begin{equation*}
\begin{split}
\Pr[\umax(\exec{})\leq B]\geq1-\negl(\secpar).
\end{split}
\end{equation*}
Recall that the analysis so far was given that $Q\geq \sqrt{n}q$. To complete the proof of the claim, we will show that for $Q<\sqrt{n}q$, the profit of the coalition $\corrupt$ cannot exceed the bound in Eq.~\eqref{eq:final_bound}, except with $\negl(\secpar)$ probability. 

If $Q<\sqrt{n}q$, then all the parties make less than $(\sqrt{n}+1)q$ random oracle queries in total per round. Let $Z$ be the number of fruits mined during the execution and $\tilde{Z}$ a random variable that follows $\msf{Bin}(\rounds(\sqrt{n}+1)q,\pf)$. By the Chernoff bounds,
\begin{equation*}
\begin{split}
\Pr\big[Z\geq\big(1+\tfrac{\log\secpar}{\sqrt[4]{\rounds n}}\big)\rounds(\sqrt{n}+1)q\pf\big]
=\negl(\secpar).  
\end{split}    
\end{equation*}
The latter implies that with $1-\negl(\secpar)$ probability the total rewards, that are clearly greater than the profit of the coalition, are no more than $\big(1+\tfrac{\log\secpar}{\sqrt[4]{\rounds n}}\big)\rounds(\sqrt{n}+1)q\pf\reward$. 

We show that if $\pf\reward>3\big(\tfrac{\cost_\msf{lc}+\cost_\msf{fs}+\cost_\msf{tx}}{(n-1)q}+\cost_\msf{ro}\big)$, then it holds that for any $\delta\in \big[\tfrac{\log\secpar}{\sqrt[4]{\rounds n}},1\big)$,
\begin{equation}\label{eq:less_nq}
\begin{split}
&\big(1+\tfrac{\log\secpar}{\sqrt[4]{\rounds n}}\big)\rounds(\sqrt{n}+1)q\pf\reward<\\
<&(1+\delta)(\rounds-1)(n-1)q\pf\reward+\log^2\secpar(n-1)q\reward-\\
&-\tfrac{n-1}{n}\big(1-\tfrac{\log\secpar}{\sqrt{\rounds}}\big)(\rounds-1)(1-(1-\pb)^{nq})\cost_\msf{lc}-\tfrac{n-1}{n}\rounds(\cost_\msf{fs}+\cost_\msf{tx})-\rounds (n-1)q\cost_\msf{ro}.
\end{split}
\end{equation}
Namely, since $\big(1+\tfrac{\log\secpar}{\sqrt[4]{\rounds n}}\big)\rounds(\sqrt{n}+1)<\tfrac{1}{2}(\rounds-1)(n-1)$ for typical values of $\rounds,n$, it holds that
\[\big(1+\tfrac{\log\secpar}{\sqrt[4]{\rounds n}}\big)\rounds(\sqrt{n}+1)q\pf\reward<\tfrac{1}{2}(1+\delta)(\rounds-1)(n-1)q\pf\reward.\]
Besides, if $\pf\reward>3\big(\tfrac{\cost_\msf{lc}+\cost_\msf{fs}+\cost_\msf{tx}}{(n-1)q}+\cost_\msf{ro}\big)$, then
\begin{equation*}
\begin{split}
&(1+\delta)(\rounds-1)(n-1)q\pf\reward+\log^2\secpar(n-1)q\reward-\\
&-\tfrac{n-1}{n}\big(1-\tfrac{\log\secpar}{\sqrt{\rounds}}\big)(\rounds-1)(1-(1-\pb)^{nq})\cost_\msf{lc}-\\
&-\tfrac{n-1}{n}\rounds(\cost_\msf{fs}+\cost_\msf{tx})-\rounds (n-1)q\cost_\msf{ro}>\\
>&(\rounds-1)(n-1)q\pf\reward-\rounds\cost_\msf{lc}-\rounds(\cost_\msf{fs}+\cost_\msf{tx})-\rounds (n-1)q\cost_\msf{ro}=\\
=&(\rounds-1)(n-1)q\pf\reward-\rounds(n-1)q\big(\tfrac{\cost_\msf{lc}+\cost_\msf{fs}+\cost_\msf{tx}}{(n-1)q}+\cost_\msf{ro}\big)>\\
>&(\rounds-1)(n-1)q\pf\reward-\tfrac{1}{3}\rounds(n-1)q\pf\reward>\\
>&(\rounds-1)(n-1)q\pf\reward-\tfrac{1}{2}(\rounds-1)(n-1)q\pf\reward=\\
=&\tfrac{1}{2}(\rounds-1)(n-1)q\pf\reward.
\end{split}
\end{equation*}
By the above, we get Eq.~\eqref{eq:less_nq}, which completes the proof of the claim.
\hfill$\dashv$
%
%
\begin{claim}\label{claim:other_deviations}
For every   $\mc{O}_\msf{tx}$-respecting adversary $\mc{A}$ that performs a combination of deviations D1, D4-D7, D9-D12 and every $\rounds$-admissible environment $\mc{Z}$ that activates the pool leader first in each round, it holds $\umax(\exec{}) \leq \umin(\honexec)$ with $1-\negl(\secpar)$ probability.
\end{claim}
\textit{Proof of Claim~\ref{claim:other_deviations}. }
D1 is captured by D6 due to step (2) (cf. Figure~\ref{fig:single_leader}) and step (2a) (cf. Figure~\ref{fig:single_other}) in the $\single$ protocol for the leader and the non leader, respectively. In more detail, if a party in $\corrupt$ does not update $inst_i$, $i\in[4]$ as instructed by $\single$ and sends its inconsistent fruits and/or blocks, then during the next round the honest parties following $\single$ will dissolve the pool. This happens because the honest parties will detect the deviation. Thus, the outcome of D1  can be captured by D6 where at some round $\mc{A}$ instructs a subset of the parties in $\mathbf{C}$ to abandon the pool. 

D12 is captured by D6 for the case where all the corrupted parties abandon the pool and follow $\protocolevp$ protocol. This happens because if the leader does not pay a non leader party, then this party will detect this via steps (2b) and (3), it will leave the pool and it will follow $\protocolevp$.  

D7 is not effective in our setting, because the parties that are not corrupted by $\mc{A}$  follow the $\single$ protocol and thus they will never produce a block $\hat{\block}:=\langle\langle \hat{h}_{-1},\hat{h}_f,\hat{\nonce},$ $\hat{\msf{dig}},\hat{\record},\hat{h}\rangle,\hat{\mbf{F}}\rangle$ or a fruit $\langle \hat{h}_{-1},\hat{h}_f,\hat{\nonce},\hat{\msf{dig}},\hat{\record},\hat{h}\rangle$ so that $(\hat{\previous},\hat{h}_f,\hat{\msf{dig}},\hat{\record})\neq(inst_1,inst_2,inst_3,inst_4)$. 

D10 is not performed by an $\mc{O}_\tx$-respecting adversary. 

D6 is captured by any combination of D2-D5 and D7-D12:  let us assume that a subset of the corrupted parties abandons the pool and creates a new pool following different instructions from the $\single$ pool. Recall that the utility of the adversary is the sum of the utilities of all the corrupted parties. Thus, the way of sharing the rewards among the corrupted parties does not affect the utility of the adversary.   

D9, D11 have the same effect as D3 for the case where $\mc{A}$ instructs all the corrupted parties to abstain by asking no queries to the random oracle $\mc{O}_\msf{ro}$. The reason is that if the pool leader does not ask the oracles $\mc{O}_\msf{fs}$, $\mc{O}_\msf{lc}$, it cannot create $inst_i$, $i\in[4]$ needed for all the parties to ask the random oracle and produce valid fruits that will give the rewards to the pool when a block is produced.  

D4 will offer to the adversary lower utility than D3 for the case where $\mc{A}$ instructs all the corrupted parties to abstain by asking no queries the random oracle $\mc{O}_\msf{ro}$. The reason that is if the adversary asks the random oracle but does not send its fruits or blocks, it incurs the cost of $C_\msf{ro}$ without getting any more rewards from the fruits it produces. 

Regarding D5, we do not consider deviations that either hinge on the assumption that the blocks can include unlimited number of fruits and/or they demand that the adversary is aware of the round when $\mc{Z}$ will terminate the execution.   
\hfill$\dashv$
\\

Since $\pf\reward>\tfrac{\cost_\msf{lc}+\cost_\msf{fs}+\cost_\msf{tx}}{(1-\frac{\log\secpar}{\sqrt[4]{n}}) \sqrt{n}q}+\cost_\msf{ro}$, $\pb=\Omega(\frac{1}{nq})$, and $\pf<\tfrac{1}{2}$, for an $\mc{O}_\msf{tx}$-respecting adversary $\mc{A}$, the conditions for all Claims~\ref{claim:H},~\ref{claim:D2_D3_D8}, and~\ref{claim:other_deviations} hold. Therefore, for any  $\mc{O}_\msf{tx}$-respecting adversary $\mc{A}$ and $\delta\in \big[\tfrac{\log\secpar}{\sqrt[4]{\rounds n}},1\big)$, with $1-\negl(\secpar)$ probability, it holds that
\begin{equation*}
\begin{split}
&\umax(\exec{})-\umin(\honexec)\leq\\
%
%
\leq&\Big(\big(\tfrac{\log\secpar}{\sqrt{\rounds n}}+\delta\big)\rounds+\log^2\secpar\big(1+\tfrac{1}{\pf}\big)-\big(\tfrac{\log^3\secpar}{\sqrt{\rounds n}}+1+\delta\big)\Big)(n-1)q\pf\reward+\\
&+\tfrac{n-1}{n}\Big((2\log\secpar)\sqrt{\rounds}+1-\tfrac{\log\secpar}{\sqrt{\rounds}}\Big)(1-(1-\pb)^{nq})\cost_\msf{lc}+\\
&+\big(1+\tfrac{\log\secpar}{\sqrt{\rounds}}\big)\rounds(1-(1-\pb)^{nq})(n-1)\cost_\msf{ltx}+\tfrac{\log^2\secpar}{n}(\cost_\msf{fs}+\cost_\msf{tx}).
\end{split}    
\end{equation*}

Thus, according to Definition~\ref{def:EVP}, the $\single$ protocol is $(n-1,0,\epsilon')$-EVP, for $\epsilon'$ as in theorem statement.

\end{proof}

\begin{remark}\label{rem:variant}
If we remove the assumption of Theorem \ref{th:equilibrium}, then we can prove that instead of the $\single$ protocol, the following strategy profile, denoted by $\mathcal{S}$, is EVP according to the utility profit: all the parties follow all the instructions of the $\single$ protocol except that:
\begin{enumerate} 
\item  the pool leader ignores step (7),(9) for all the rounds. \item in step (8), if the round is a \textit{payment round}, the pool leader sets $inst_4\leftarrow \tx_T$, where $\tx_T$ is the special transaction with the payments, otherwise it does nothing.
\item the members in step (3) do not add $C_\msf{tx}$ in cost.
\end{enumerate}
\end{remark}

\section{Discussion}\label{sec:discussion}

 We believe that proposals such as Smartpool~\cite{10.5555/3241189.3241299} that give the transaction verification back to the miners, and \cite{236314} where the miners validate the transactions, have similar tendency to centralization; like in the FruitChain system, the miners in~\cite{10.5555/3241189.3241299,236314} have incentives to collude in order to share the transaction verification costs. Note that as  \cite{Azouvi2021SoK} states, any decentralised system can be executed in a centralised manner. Thus, the fact that a system is designed so that miners can process the transactions in a decentralised manner does not imply that they have incentives to do so.
\par As our results indicate, apart from reducing the variance of the rewards, further research is needed to incentivize decentralization in PoW protocols. One possible research direction is to disincentivize parallel mining (e.g., \cite{9136680,10.1145/2810103.2813621}).
\par In more detail, Miller et al. \cite{10.1145/2810103.2813621} propose a non-outsourceable puzzle that does not allow miners in a pool to provide the leader with a proof that they indeed mine for the pool. This could constitute a possible countermeasure for PoW systems where centralization is motivated by sharing transaction verification costs. The reason is that this non-outsourceable puzzle could render abandoning the pool more profitable than sticking to the pool and share the costs. However, such countermeasure seems incompatible with any blockchain protocol that, like FruitChain, uses the 2-for-1 PoW technique \cite{backbone} to reduce the variance of the rewards and mitigate selfish mining attacks. This is because the ``easier'' puzzle used in the 2-for-1 PoW technique  can serve as proof of mining in a pool.
\par  To overcome the above incompatibility, the design challenge is to construct PoW puzzles that disincentivize the formation of pools, while being applicable to blockchain protocols that satisfy fairness~\cite{fruitchain}.
One candidate solution to this direction could be built upon the PoW puzzle in \cite{9136680} that (i) is non-parallelizable, (i.e., it is computed by the miners serially), and (ii) is used in a consensus mechanism that guarantees fairness.

Note that in order to eliminate the problem of centralization, it is necessary that the miners do not use the same set of transactions in their puzzles. Else, they will still have incentives to collude and share the transaction verification cost. 

\paragraph{Acknowledgements.}
Zacharias was supported by Input Output (\url{https://iohk.io}) through their funding of the Edinburgh Blockchain Technology Lab. Part of this work was conducted while Stouka was a research associate at the Edinburgh Blockchain Technology Lab.

\bibliographystyle{plain}
\bibliography{references}
\appendix

\renewcommand{\thesection}{\Alph{section}}
\section{Chernoff bounds}\label{app:chernoff}
 Let $X_1,X_2,\ldots,X_K$ be independent random variables, where $X_i=1$ with probability $p_i$ and $X_i=0$ with probability $(1-p_i)$. Let $X=\sum_{i=1}^KX_i$ so $E[X]=\sum_{i=1}^Kp_i$. We apply the Chernoff bounds in the following form: For any $0<\delta<1$,
\begin{align*}
\Pr[X\geq(1+\delta)E[X]]&\leq e^{-\frac{\delta^2}{3}E[X]}\\
\Pr[X\leq(1-\delta)E[X]]&\leq e^{-\frac{\delta^2}{2}E[X]}
\end{align*}

\end{document}